\newif\ifreport\reporttrue

\ifreport
\documentclass[journal,12pt,onecolumn,draftclsnofoot,]{IEEEtran}
\else
\documentclass{sig-alternate-05-2015}

\fi

\usepackage{amsfonts}
\usepackage{mathrsfs}
\usepackage{graphicx,epsfig,amssymb,amsmath,url,latexsym}
\usepackage{cases}
\makeatletter
\newif\if@restonecol
\makeatother

\usepackage{subfigure}
\usepackage[square,sort&compress,numbers]{natbib}
\usepackage{bm}
\usepackage{multirow}
\usepackage{color}
\usepackage{enumerate}
\usepackage{setspace}
\usepackage{graphicx}
\usepackage{bbm}
\usepackage{url}
\usepackage[utf8]{inputenc}
\usepackage[T1]{fontenc}

\usepackage{hyperref}

\usepackage{algorithm} 
\usepackage[noend]{algorithmic} 
\usepackage{multirow} 
\usepackage{amsmath} 
\usepackage{xcolor}

\newtheorem{theorem}{Theorem}
\newtheorem{lemma}{Lemma}
\newtheorem{proposition}{Proposition}
\newtheorem{corollary}{Corollary}
\newtheorem{definition}{Definition}

\newtheorem{remark}{Remark}[theorem]

\DeclareMathOperator*{\argmin}{arg\,min}
\DeclareMathOperator*{\argmax}{arg\,max}

\begin{document}

\title{Concurrent Channel Probing and Data Transmission in Full-duplex MIMO Systems}

%
\ifreport
\author{
Zhenzhi Qian, Fei Wu, Zizhan Zheng, Kannan Srinivasan, and Ness B. Shroff

\thanks{Zhenzhi Qian, Fei Wu and Kannan Srinivasan are with the Department of CSE, The Ohio State University, Columbus, OH, 43210 (e-mail: qian.209@osu.edu, wu.1973@osu.edu, kannan@cse.ohio-state.edu).}
\thanks{Zizhan Zheng is with the Department of Computer Science, Tulane University, New Orleans, LA 70118. (e-mail: zzheng3@tulane.edu).}
\thanks{Ness B. Shroff is with the Departments of ECE and CSE, The Ohio State University, Columbus, OH, 43210 (e-mail: shroff.11@osu.edu).}
}
\fi


\maketitle


\begin{abstract}

An essential step for achieving multiplexing gain in MIMO downlink systems is to collect accurate channel state information (CSI) from the users. Traditionally, CSIs have to be collected before any data can be transmitted. Such a sequential scheme incurs a large feedback overhead, which substantially limits the multiplexing gain especially in a network with a large number of users. In this paper, we propose a novel approach to mitigate the feedback overhead by leveraging the recently developed Full-duplex radios. Our approach is based on the key observation that using Full-duplex radios, when the base-station (BS) is collecting CSI of one user through the uplink channel, it can use the downlink channel to simultaneously transmit data to other (non-interfering) users for which CSIs are already known. By allowing concurrent channel probing and data transmission, our scheme can potentially achieve a higher throughput compared to traditional schemes using Half-duplex radios. The new flexibility introduced by our scheme, however, also leads to fundamental challenges in achieving throughout optimal scheduling. In this paper, we make an initial effort to this important problem by considering a simplified group interference model. We develop a throughput optimal scheduling policy with complexity $O((N/I)^I)$, where $N$ is the number of users and $I$ is the number of user groups. To further reduce the complexity, we propose a greedy policy with complexity $O(N\log N)$ that not only achieves at least 2/3 of the optimal throughput region, but also outperforms any feasible Half-duplex solutions. We derive the throughput gain offered by Full-duplex under different system parameters and 
show the advantage of our algorithms through numerical studies.
	
\end{abstract}

\section{Introduction}
\label{sec:introduction}
Mobile data traffic is expected to increase at rate of $53\%$ per year by 2020 \cite{cisco}. Multi-user MIMO (MU-MIMO), which can potentially increase the network capacity linearly with the number of users, has been considered as an important technique to confront this data traffic challenge. Theoretically, in a system with $M$ transmit and receive antennas, the throughput using MU-MIMO can be $M$ times of the throughput using a single transmit and receive antenna pair \cite{tse2005fundamentals}, where $M$ is commonly referred as the spatial multiplexing gain. 

In this paper, we consider one important application of MU-MIMO, i.e., the downlink wireless cellular network consisting of one Base Station (BS) equipped with many antennas and many users each equipped with one antenna. In such systems, the BS could utilize MU-MIMO to transmit multiple data streams to multiple users simultaneously. Nevertheless, to take the advantage of MU-MIMO in practice, it is prerequisite for the transmitter to learn the accurate channel state information (CSI) of the users \cite{liu2016understanding}. Note that in traditional wireless networks, radios can only operate in Half-duplex (HD) mode, i.e., a radio cannot transmit and receive packets on the same frequency at the same time. As a result, traditional schemes to harness the multiplexing gain of MU-MIMO, e.g., \cite{introduction04,zhou2015signpost}, requrie that the channel state information (CSI) of the users have to be learned first before any data can be transmitted. Such a sequential channel learning scheme incurs a large overhead when there are a large number of users, which would in turn substantially limit the multiplexing gains of MU-MIMO, especially if the channel coherence time is relatively short \cite{zhou2015signpost,introduction04},  {\it The large channel learning overhead has been a long-standing open problem which limits the achievable throughput of MU-MIMO in practice.}



Recently, Full-duplex (FD) radios  \cite{choi2010achieving, duarte2012experiment, bharadia2013full} have been developed, which allow simultaneous transmission and reception on the same frequency. The availability of Full-duplex provides significant flexibility in designing wireless resource allocation algorithms. For example, it has been shown that in some cases \cite{yang2014characterizing}, Full-duplex can almost double the throughput and effectively improve spectrum efficiency. This leads to the following natural and important question: \emph{Is it possible to leverage Full-duplex to address the feedback overhead challenge in Multi-user MIMO downlink systems?}

In this paper, we answer this question in the affirmative. 
By using a Full-duplex BS, we are able to break the boundary between the channel learning phase and the data transmission phase. As shown in Fig. \ref{fig:scheme}, the BS receives the channel probing signal from Alice in round 1 and measures the downlink channel to Alice assuming channel is reciprocal\footnote{Measuring downlink channel to a user through channel probing from the user is standard in a time division duplex (TDD) system \cite{introduction04,zhou2015signpost}.}. Then in round 2, the BS uses Full-duplex capability to send data to Alice and receive the probing signal from Bob simultaneously, assuming Bob does not interfere with Alice. After the BS measures all downlink channels, the BS operates in MU-MIMO mode in round 3. Compared to Half-duplex systems, once the BS knows the downlink channel to Alice, it can start transmission immediately rather than waiting until the end of the channel learning phase. 
Henceforth, we will refer to this concept as {\it concurrent channel probing and data transmission}. 

\begin{figure}[htbp]
	\centering
	\includegraphics[width=4in]{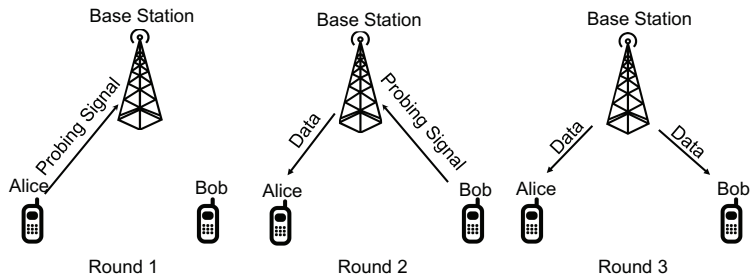}
	\vspace{-0.5cm}
	\caption{Concurrent channel probing and data transmission.}\label{fig:scheme}
\end{figure}

Due to the interference between users, the performance of concurrent channel probing and data transmission scheme depends highly on the set of users selected to send probing signals and the ordering of these users. Therefore, the following important question remains: \emph{How do we design a low-complexity scheduling policy that achieves provably good throughput performance under the concurrent channel probing and data transmission?} 

While the design of high performance scheduling policies have been extensively studied in traditional wireless systems \cite{lin2006tutorial}, relatively few efforts \cite{yang2015scheduling} have focused on the scheduling problem in Full-duplex systems. In particular, it is much more challenging to consider this problem under concurrent channel probing and data transmission. The reason is that: 1) The ordering of users sending probing signal matters. A user that sends a probing signal earlier also starts transmission earlier. 2) Within one channel coherence time, the scheduling decisions are coupled in terms of time and interference relations. The rate received by a certain user depends on what time it transmits the probing signal as well as the interference relations with the users scheduled to send probing signals later. These two facts make the scheduling problem more complicated and classical scheduling policies do not apply here.
In this paper, we aim to develop a throughput near-optimal scheduling policy and investigate the Full-duplex gain for a various of network settings. 

The key contributions of this paper are summarized as follows:
\begin{itemize}
	\item We develop a scheduling policy that achieves the optimal throughput region under concurrent channel probing and data transmission. Compared to Brute-Force search, the complexity has been decreased from $O(N!)$ to $O((N/I)^I)$.
	\item To further reduce the scheduling complexity in large systems, we design a greedy policy with complexity $O(N\log N)$ that not only achieves at least 2/3 of the optimal throughput region but also outperforms any feasible Half-duplex solutions. We conjecture that the real performance of the greedy policy is very close to the optimal, which is confirmed by simulations. 
	\item We derive the Full-duplex gain under different system parameters and use simulations to validate our theoretical results. 
\end{itemize}


The rest of the paper is organized as follows. We discuss related works in Section \ref{sec:related}. In Section \ref{sec:system}, we describe the system model and problem formulation. In Section \ref{sec:opt}, we develop a throughput optimal policy which stabilizes the system under any feasible arrival rates. In Section \ref{sec:greedy}, we design a low-complexity greedy policy and provide provable performance guarantees. In Section \ref{sec:discuss}, we derive the Full-duplex gain of different network settings and system parameters. We conduct simulations to validate our theoretical results in Section \ref{sec:sim} and make concluding remarks in Section \ref{sec:conclusion}.

\section{Related Work}\label{sec:related}
In-band Full-duplex, as an emerging technology in wireless communication, was implemented by combining RF and baseband interference cancellation \cite{choi2010achieving, duarte2012experiment, bharadia2013full}, enabling simultaneous bi-directional transmission between a pair of nodes. Full-duplex has now been widely studied in a number of wireless communication scenarios. Full-duplex WiFi-PHY based MIMO radios was first implemented in \cite{bharadia2014full}, and experiments showed that the theoretical doubling of throughput is practically achieved.
While it is hard to make Full-duplex MIMO radios fit in small personal devices, it is feasible to build a Full-duplex MIMO Base Station due to bigger size and more powerful computational ability \cite{everett2016spatial}. In \cite{du2014mimo, du2015mu}, the authors proposed the continuous feedback channel, which enables sequential beamforming that update weights while also performing downlink transmission. The authors showed that the system outperforms its Half-duplex counterpart  and reduced the control overhead at the same time. This work can be viewed as an preliminary attempt of the idea of concurrent channel probing and data transmission. However, the authors assumed that users are symmetric and did not consider the scheduling problem, which is the focus of our study here.

In addition to the research efforts focused on implementation and experiments, there have also been several theoretical works on Full-duplex systems. Although Full-duplex is expected to double the capacity in single pair of nodes, \cite{xie2014does} showed that the inter-link interference and spatial reuse substantially reduces network-level Full-duplex gain, making it less than 2 in typical cases. In order to deal with the increasing inter-link interference, \cite{sahai2013uplink} presented a new interference management strategy to achieve a larger rate gain over Half-duplex systems. The capacity region of multi-channel Full-duplex links was characterized in \cite{maravsevic2016capacity} and rate gain is illustrated for various channel and cancellation scenarios. The authors in \cite{yang2014characterizing} also investigated the achievable throughput performance of MIMO, Full-duplex and their variants that allow simultaneous activation of two RF chains. The scheduling problem in Full-duplex cut-through transmission was considered in \cite{yang2015scheduling}, where the authors characterized the interference relationship between links in the network with cut-through transmission and designed a Q-CSMA type of scheduling algorithm to leverage the flexibility of Full-duplex cut-through transmission. In contrast to the aforementioned works, this is the first work
that considers the scheduling problem under concurrent channel probing and data transmission and provides analytical framework to characterize the network-level Full-duplex gain.

\section{System model}\label{sec:system}
We consider the downlink phase of a single-cell Full-duplex MIMO system. There are $N$ users in this system and each of them is equipped with only one antenna. The Base Station (BS) has multiple antennas and Full-duplex capability. In addition, we assume time is slotted and we consider a discrete-time system. We use $\mathcal{N}$ denote the set of all users in the system.
\subsection{Channel Model}
We consider a block fading channel, where the channel state remains the same within each time-slot, but may vary from time-slot to time-slot. We assume channel state information (CSI) is only available at the user side at the beginning of each time-slot. In order to fully achieve the multiplexing gain of MU-MIMO, the BS needs to collect CSI via feedback through the uplink channel. We assume that channels are reciprocal, in which case a user could send a probing signal on its single antenna and the BS, by measuring on its antennas, learns the downlink CSI.  Any CSI expires by the end of the current time-slot, and it has to be learned again in the next time-slot. In practice, collecting CSI from multiple users takes time and its overhead is linear with respect to the number of the corresponding users. We assume that in one time-slot, the transmitter can collect CSI from at most $K$ users. 
Therefore, each time-slot can be further divided into $K$ mini-slots and it takes one mini-slot to learn each CSI. The BS can only transmit one packet per mini-slot to each user whose channel information is already known. 

In traditional Half-duplex systems, CSI collection and data transmission must be separated in time to avoid interference. Data transmission phase starts only if all desired CSIs are collected. 
Full-duplex systems, on the other hand, allows data transmission immediately after each CSI is collected.
\subsection{User Groups}
Full-duplex capability does not always offer ``free lunch'', its performance suffers from complex interference patterns. One way to characterize interference is using user groups which guarantee no inter-group interference. Thus, we can break the scheduling problem into two steps:
1) Given $N$ users, how to divide them into different user groups. 2) Given group information, how to find a scheduling policy that achieves good throughput performance. Dividing users into groups is not easy due to the conflict between interference constraints and the desire to have more groups and less users in each group. We focus on the second step in this work and leave the joint problem as the future work. The problem is still challenging even when the group information is already given.

Assume $N$ users are split into $I$ user groups, which guarantees no inter-group interference. For example, suppose user $u_i$ and $u_j$ are from different groups, the uplink stream of user $u_i$ does not interfere with the downlink stream of user $u_j$. 
Based on each user's geographical statistics, the group information will be determined once over a much larger time scale. The group information is assumed to be static and remains the same from time-slot to time-slot. Fig. \ref{fig:system} is an illustration of a downlink system with 2 user groups. We use $g(u)$ to denote the group index of user $u$, and let $\mathcal{G}_{g(u)}$ denote the set of users in group $g(u)$.

\begin{figure}[htbp]
	\centering
	\includegraphics[width=2.8in]{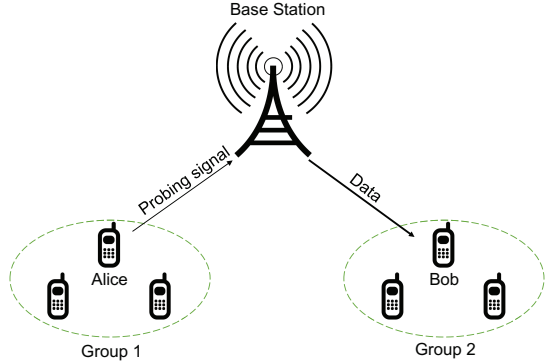}
	\caption{A downlink system with 2 user groups, the BS receives probing signal from Alice and transmits data packets to Bob (channel is already known) simultaneously.}\label{fig:system}
	\vspace{-0.4cm}
\end{figure}

\subsection{Traffic Model}
The BS maintains a queue $Q_u$ to store packets requested by each user $u$. The arrival process to each queue is assumed to be stationary
and ergodic. We assume packet arrival and departure both occur at the beginning of each time-slot. Let $A_u[t]$ denote
the number of packet arrivals to queue $Q_u$ in time-slot $t$. Let $R_u[t]$ denote the downlink rate to queue $Q_u$ in time-slot $t$. The queue-length $Q_u[t]$ evolves as:
\begin{align}\label{Q_evo}
Q_u[t+1]=\max\left\{Q_u[t]+A_u[t]-R_u[t], 0\right\}.
\end{align}
\subsection{Scheduling Policy}
In each time-slot $t$, a scheduling policy $P$ determines the schedule based on the system state, e.g., queue-length and delay. Such schedule can be described as a scheduling vector $\mathbf{f}=(u_1,\cdots, u_K)$, which indicates that user $u_i$ sends a probing signal in the $i^{th}$ mini-slot. $u_i=0$
implies that the BS is only transmitting, not learning any channel in the $i^{th}$ mini-slot. ``0'' element is also considered as a dummy user from a dummy group with zero queue-length. Due to interference constraints, 
once the BS chooses to learn user $u$'s channel during the $i^{th}$ mini-slot, it will block all other users in $\mathcal{G}_{g(u)}$ from receiving any packet. However, the BS can transmit data packets to users from other groups since there is no interference between these groups. 
We use $R_{u_i}^{\mathbf{f}}$ to denote the downlink rate to user $u_i$ under scheduling vector $\mathbf{f}$. For all $i=1, \dots, K$, $R_{u_i}[t]=R_{u_i}^{\mathbf{f}}$ if scheduling vector $\mathbf{f}$ is adopted in time-slot $t$. From now on, we omit the subscript $[t]$ when looking into the schedule made in a certain time-slot $t$. Note that $R_{u_i}^{\mathbf{f}}$ is the number of mini-slots from $i+1$ to $K$ such that the group of the scheduled user is different from group $g(u_i)$, i.e., $R_{u_i}^{\mathbf{f}}=\sum_{j=i+1}^{K}\mathbf{1}_{\{g(u_i)\neq g(u_j)\}}$.
For example, if $\mathbf{f}=(u_a, u_b, u_c,0,\cdots,0)$ and $g(u_a)=g(u_b)\neq g(u_c)$. 
From the second mini-slot to the $K^{th}$ mini-slot, there are $K-2$ users in $\mathbf{f}$ such that its group is other than $g(u_a)$. Thus, $R_{u_a}^{\mathbf{f}}=K-2$. Similarly, we have $R_{u_b}^{\mathbf{f}}=K-2$ and $R_{u_c}^{\mathbf{f}}=K-3$.
Denote the set of feasible scheduling policies as $\Pi$.

In this paper, we mainly focus on the throughput performance of the system. First we define the \emph{optimal throughput region} for any given system parameters $N$ and $K$.
As in \cite{andrews2004scheduling, ji2015achieving}, a stochastic queueing network is said to be stable if behaves as a discrete-time countable Markov chain and the Markov chain is \emph{stable} in the following sense: 1) The set of positive recurrent states is non-empty. 2) It contains a finite subset such that with probability one, this subset is reached within finite time from any initial state. When all the states communicate, stability is equivalent to the Markov chain being positive recurrent \cite{neely2013delay}. 
The \emph{throughput region} $\Lambda^{P}$ of a scheduling policy $P$ is defined as the set of arrival rate vectors for which the network remains stable under this policy. 
\begin{definition}
	(\textit{Optimal throughput region})
	The optimal throughput region is defined as the union of the throughput regions of all possible scheduling policies, which is denoted by $\Lambda^*$, i.e., 
	\begin{align}
	\Lambda^*=\bigcup_{P\in\Pi}\Lambda^{P}.
	\end{align}
\end{definition}
\begin{definition}
	(\textit{Throughput optimal policy})
	A scheduling policy is throughput-optimal if it can stabilize any arrival rate vector strictly inside $\Lambda^*$.
\end{definition}



\section{Optimal Scheduling Policy}\label{sec:opt}
In this section, we propose a throughput-optimal scheduling policy to the concurrent probing and transmission problem. We first observe that the following classic result applies to our setting as well.

\begin{theorem}
	Any policy that maximizes the weight $w(\mathbf{f})=\sum\limits_{u\in\mathcal{N}} Q_u R_{u}^{\mathbf{f}}$ in each time-slot, a.k.a., the MaxWeight scheduling policy, is throughput-optimal.\label{thm:M}
\end{theorem}
\begin{proof}
	Please refer to the proof in \cite{tassiulas1992stability}.
\end{proof}

From the theorem, it suffices to find a scheduling vector $\mathbf{f}^{*}$ such that the weight $w(\mathbf{f})$ is maximized in each time-slot, i.e.,
\begin{align}
\mathbf{f}^*=\argmax_{\mathbf{f}}\sum_{u\in\mathcal{N}} Q_u R^{\mathbf{f}}_{u}.
\end{align}

However, it is not trivial to find a MaxWeight schedule with low complexity. We note that for traditional wireless scheduling under 1-hop interference, MaxWeight scheduling boils down to finding a maximum weighted matching in each time-slot, which can be done in $O(N^3)$ where $N$ is the number of nodes. This result does not apply to our setting, however, since the ordering of users sending probing signal matters. A Brute-Force search enumerates all possible permutations of users, leading to a high complexity of $O(N!)$, which is infeasible when $N$ is large. Thus, an interesting question is how to find a MaxWeight schedule in our setting in a more efficient way. To this end, we propose the following algorithm with complexity $O((N/I)^I)$ (polynomial when $I$ is a constant regardless of $N$). In the algorithm, $m_i$ indicates the number of users to be chosen from group $i$, $1\leq i\leq I$, and $\mathbf{m} = (m_1,\cdots, m_I)$ is the user-selection vector. Algorithm~\ref{alg1} will be applied to each time-slot to generate the MaxWeight schedule.

\begin{algorithm} %
	\caption{Search algorithm for MaxWeight Schedule} 
	\label{alg1} 
	\begin{algorithmic}[1] 
		\REQUIRE For all $u\in \mathcal{N}$, group $g(u)$ and queue-length $Q_u$.  
		\ENSURE Scheduling vector $\mathbf{\widehat{\mathbf{f}}}$
		
		\STATE Initialization: User-selection vector $\mathbf{m}=(0,0,\cdots,0)$, \\
		$\widehat{w}=0$, $\widehat{\mathbf{f}}=(0,0,\cdots,0)$.
		\FORALL{$\mathbf{m}$ such that $\sum_{i}m_i \leq K$}
		\STATE Set scheduling vector $\mathbf{f}=(0,0,\cdots,0)$.
		\STATE Set scheduled user set $U=\emptyset$
		\FOR {i=$1, 2,\cdots,I$}
		\STATE  Add $m_i$ users with longest queue-length from group $i$ to $U$.
		\ENDFOR
		\STATE Fill in scheduling vector $\mathbf{f}$ with users in $U$, \\following the Longest Queue-length First order.
		\IF{$w(\mathbf{f})>\widehat{w}$}
		\STATE $\widehat{w}=w(\mathbf{f})$
		\STATE $\widehat{\mathbf{f}}=\mathbf{f}$
		\ENDIF
		\ENDFOR
		\RETURN{$\widehat{\mathbf{f}}$}
		
	\end{algorithmic}
\end{algorithm}
\vspace{-0.2cm}

For a given user-selection vector $\mathbf{m}$, Algorithm~\ref{alg1} picks $m_i$ users from group $i$ with longest queue-length, for all $i=1,2,\cdots, I$. It then generates a candidate scheduling vector $\mathbf{f}$ by filling in users following the Longest Queue-length First (LQF) order. The weight $w(\mathbf{f})$ is evaluated for all possible user-selection vectors $\mathbf{m}$ and its resulting scheduling vector, Algorithm~\ref{alg1} returns the scheduling vector $\widehat{\mathbf{f}}$ that has the maximum weight.

\begin{theorem}
	The schedule $\widehat{\mathbf{f}}$ returned by Algorithm~\ref{alg1} maximizes weight $w(\mathbf{f})$.
\end{theorem}

\begin{proof}
	We divided the proof into two steps. For the first step, we show that the LQF maximizes the weight for a given scheduled user set. Then for the user-selection part, we show that it is sufficient to evaluate all possible user-selection vectors $\mathbf{m}$ and its resulting scheduled user set by adding $m_i$ users with longest queue-length from each group $i$.
	We first present several properties of MaxWeight schedule that will be used later.
	\begin{lemma}\label{lem:shift_zero}
		For any scheduling vector with ``0'' element(s) between two adjacent non-zero elements, the total weight will not decrease by shifting the ``0'' element to the end, i.e., there is no ``idle'' (not learning any user's channel) mini-slot in between two ``busy'' mini-slots.
	\end{lemma}
	\begin{proof}
		Please see APPENDIX \ref{APP:L1}.
	\end{proof}
	
	\begin{corollary}
		The optimal scheduling vector must take the form $\mathbf{f^*}=(u_1, u_2, \cdots, u_{\Omega}, 0, 0, \cdots, 0)$, where $u_1, \cdots, u_\Omega$ are non-zero and $\Omega< \min\{K,N\}$.
	\end{corollary}

	\begin{remark}
	It is also challenging to determine the optimal value of $\Omega$, which depends on group settings as well as instantaneous queue-length. 
	\end{remark}
	
	\begin{lemma}
		For any scheduling vector $\mathbf{f}=(u_1,\cdots, u_\Omega, 0, \cdots, 0)$, the total weight $w(\mathbf{f})$ will not decrease by reordering the users following queue-length descending order (longest queue-length first, LQF).
		\label{lem:LQF}
	\end{lemma}
	
	\begin{proof}
		Please see APPENDIX \ref{APP:L2}.
	\end{proof}
	
	From Lemma \ref{lem:shift_zero} and Lemma \ref{lem:LQF}, we know that for a fixed scheduled user set $\{u_1, u_2, u_3, \cdots, u_{\Omega}\}$ with $Q_{u_1}\geq Q_{u_2}\geq\cdots\geq Q_{u_{\Omega}}$, the optimal schedule $\mathbf{f}^*$ takes the form $(u_1, \cdots, $
	$u_{\Omega}, 0, \cdots, 0)$. From now on, for a given scheduled user set, we only need to focus on the LQF schedule.
	\begin{remark}
		Lemma \ref{lem:LQF} holds only for a given scheduled user set, applying LQF to the set of all users does not guarantee the maximum. Since LQF is a myopic rule, it always gives higher priority to users with longer queue-length regardless of their interference relations. In fact, queue-length and interference relations both play a key role in this problem, and we need to do user-selection to get a good balance between these two factors.
	\end{remark}
	
	For the second step, we will focus on the user-selection part. For a given user-selection vector $\mathbf{m}$, we want to show that choosing $m_i$ users with the longest queue-length from each group $i$ is the best option to maximize weight.
	Denote $\mathcal{P}^{\mathbf{m}}_i$ to be the set of users from group $i$ with $m_i$ longest queue-length, $\mathcal{U}^{\mathbf{f}}_i$ to be the set of users from group $i$ that are selected by schedule $\mathbf{f}$, we have the following lemma.
	
	\begin{lemma}\label{lem:chooselargest}
		Consider a given user-selection vector $\mathbf{m}$, and choose an arbitrary LQF schedule $\mathbf{f}$.
		Pick user $u_s$ with the longest queue-length in the set $\mathcal{U}^{\mathbf{f}}_i/\ \mathcal{P}^{\mathbf{m}}_i$ (if it is not empty), and replace it by user $u_l$ that has the longest queue-length in the set $\mathcal{P}^{\mathbf{m}}_i/\ \mathcal{U}^{\mathbf{f}}_i$. Denote the new LQF schedule as $\mathbf{f}'$, we have $w(\mathbf{f}')\geq w(\mathbf{f})$.
	\end{lemma}
	
	\begin{proof}
		Please see APPENDIX \ref{APP:L3}.
	\end{proof}
	\begin{remark}
		The equality in Lemma \ref{lem:chooselargest} holds if and only if the queue-lengths of $u_s$ and $u_l$ are the same.
	\end{remark}
	\begin{lemma}\label{lem:chooselargest2}
		Given any user-selection vector $\mathbf{m}$, any LQF schedule $\mathbf{f}$ maximizes weight $w(\mathbf{f})$ must pick $m_i$ users with longest queue-length in each group $i$ for any $i=1,2,\cdots,I$.
	\end{lemma}
	
	\begin{proof}
		Please see APPENDIX \ref{APP:L4}.
	\end{proof}
	From Lemma \ref{lem:chooselargest2} we know, given user-selection vector $\mathbf{m}$, the best schedule will always pick $m_i$ users with longest queue-length from each group $i$ for any $i=1,2,\cdots,I$. In addition, the best ordering of these users will be the LQF order. Therefore, given $\mathbf{m}$, the schedule yields maximum weight is determined by:
	(1) For each group $i$, add $m_i$ users with longest queue-length into the scheduled user set $U(\mathbf{m})$.
	(2) Schedule the users from $U(\mathbf{m})$ following the LQF order. Thus, traversing all possible $\mathbf{m}$ will return the MaxWeight schedule. And this proves the optimality of Algorithm \ref{alg1}.
\end{proof}

\section{A Low-complexity Greedy Policy}\label{sec:greedy}
Although Algorithm \ref{alg1} returns throughput optimal policy in polynomial time, the complexity $O((N/I)^I)$ grows very high when the number of groups $I$ is large. 
It is interesting to see whether there is any low-complexity policy that achieves provably good throughput. In this section, we propose a greedy algorithm which incrementally adds users to the schedule and prove that it achieves at least 2/3 of the optimal throughput region. In addition, our proposed greedy policy always achieves a larger throughput region than any scheduling policies under Half-duplex.
\subsection{Greedy Algorithm Description}
\begin{definition}
	(\textit{Marginal Gain}) Given a schedule $\mathbf{f}=(u_1,\cdots, u_\Omega, 0,\cdots, 0)$ and a user $u$ that is a candidate user to be considered in $j^{th}$ mini-slot (when evaluating user $u$, the first $j-1$ scheduled users have already been determined in $\mathbf{f}$), the marginal gain ${\Delta}^{\mathbf{f}, j}_{u}$ is defined to be the weight difference caused by adding user $u$ as the $j^{th}$ element of $\mathbf{f}$, assuming there are no future scheduled users, i.e., ${\Delta}^{\mathbf{f}, j}_{u}=w\left((u_1,\cdots, u_{j-1}, u, 0, \cdots, 0)\right)-w\left((u_1, \cdots, u_{j-1}, 0, \cdots, 0)\right)$. 
\end{definition}

To evaluate the marginal gain of adding user $u$ to the schedule $\mathbf{f}$, we must consider the benefit as well as the cost. The benefit is obvious, we have one more user and it keeps transmitting packets until the end of the current time-slot, i.e., receives a rate of $K-j$. Hence its weight contribution is $Q_u(K-j)$. On the other hand, if we schedule user $u$ in $j^{th}$ mini-slot, it will block the transmission of the previously scheduled users that are from the same group $g(u)$. Thus, the weight loss is $\sum_{i=1}^{j-1}Q_{u_i}\mathbf{1}_{\{g(u_i)=g(u)\}}$. Therefore, we have:
\begin{align}
{\Delta}^{\mathbf{f}, j}_{u}=Q_{u}(K-j)-\sum\limits_{i=1}^{j-1}Q_{u_i}\mathbf{1}_{\{g(u_i)=g(u)\}}.
\end{align}

A positive marginal gain means that by adding a new user, the weight will not be decreased. Marginal gain considers queue-length as well as the group information and is able to discriminate different cases (e.g., long queue-length \& strong interference v.s. short queue-length \& weak interference). Although the marginal gain is not the actual gain of user $u_j$ since we do not know the future scheduled users, it is still a good metric to evaluate the potential gain of adding one candidate user to the current schedule. Moreover, as we will soon see, the Marginal Gain-based Greedy (MGG) Algorithm achieves good throughput performance. 

The MGG Algorithm, inspired by Section \ref{sec:opt}, we first sort users according to their queue-lengths, and then start from the user that has the longest queue-length in the system, the MGG Algorithm iteratively evaluates the user $u$ with next longest queue-length. The MGG Algorithm will add user $u$ if its marginal gain is positive, otherwise skip user $u$ and continue to evaluate the user with the next longest queue-length until $K$ users have been scheduled or all $N$ users are all evaluated. 
\begin{algorithm} %
	\caption{Marginal Gain-based Greedy Algorithm} 
	\label{alg2} 
	\begin{algorithmic}[1] 
		\REQUIRE $\forall$ user $u\in\mathcal{N}$, group $g(u)$ and queue-length $Q_u$.  
		\ENSURE Scheduling vector $\mathbf{f}^G$
		
		\STATE Initialization: $\mathbf{f}^G=(0,\cdots,0)$ 
		\STATE Initialization: $index=1$
		\STATE Sort queue-length, assume $Q_{u_1}\geq Q_{u_2}\geq\cdots Q_{u_N}$
		\FORALL{$i$ from $1$ to $N$}
		\IF{$index\leq K$}
		\IF{${\Delta}^{\mathbf{f}^G, index}_{u_i}\geq 0$}
		\STATE Add user $u_i$ to $\mathbf{f}^G$ as the index$^{th}$ element
		\STATE $index=index+1$
		\ENDIF
		\ENDIF
		\ENDFOR
		\RETURN{$\mathbf{f}^G$}
		
	\end{algorithmic} 
\end{algorithm}

The complexity of Algorithm \ref{alg2} is at most $O(N\log N)$ (comes from the sorting operation), regardless of the value $I$ takes. Compared to Algorithm \ref{alg1}, Algorithm \ref{alg2} uses LQF and marginal gain to efficiently select valuable users. Again, applying LQF only would work poorly, since it only gives higher priority to those users with longer queue-length rather than large marginal gain. In fact, the inter-user interference is very important and should not be ignored.
\subsection{Performance Analysis} 
The MGG Algorithm is simple, however it sacrifices some throughput performance. In this section, we aim to provide a theoretical worst-case lower bound on its throughput performance.
\begin{theorem}
	The Greedy Algorithm \ref{alg2} stabilizes at least 2/3-fraction of the arrival vector on the optimal throughput region. (Achieves 2/3 of the optimal throughput region).
\end{theorem}
\begin{proof}
	From \cite{eryilmaz2005stable}, we know that it suffices to show that $w(\mathbf{f}^G)\geq 2/3w(\mathbf{f}^*)$, where $\mathbf{f}^*$ is the MaxWeight schedule. Consider the users selected by $\mathbf{f}^G$ and $\mathbf{f}^*$. Let $\mathcal{A}$ denote the set of users shared by both schedules, let $\mathcal{B}$
	denote the set of users only scheduled in $\mathbf{f}^*$ and let $\mathcal{C}$ denote the set of users only scheduled in $\mathbf{f}^{G}$. 
	\begin{remark}
		The MaxWeight schedule is not necessarily unique, but these schedules have the same weight. We can choose any of these schedules to be schedule $\mathbf{f}^*$ here. 
	\end{remark} 
	\begin{remark}
		In practice, users in $\mathcal{B}$ could interfere with users in $\mathcal{A}$. Here in this proof, we aim to show a stronger claim which assumes that in the MaxWeight schedule, users from $\mathcal{B}$ do not interfere with users in $\mathcal{A}$ and $\mathcal{B}$ itself.
	\end{remark}
	\begin{definition}
		(\textit{Extra weight})
		Extra weight $\epsilon$ is defined to be the weight loss in the MGG schedule caused by interference from users in $\mathcal{C}$. That is to say, the total weight $w(\mathbf{f}^G)+\epsilon$ is calculated as if there is no interference caused by users in $\mathcal{C}$, adding each user in $\mathcal{C}$ does not block the downlink transmission of all the scheduled users which are from the same group.
	\end{definition}
	
	We divide the proof into two parts, for the first part, we show that $w(\mathbf{f}^G)+\epsilon \geq w(\mathbf{f}^*)$. Then we show that $\epsilon \leq 1/2 w(\mathbf{f}^{G})$. Combining both parts, we know $w(\mathbf{f}^{G})\geq 2/3w(\mathbf{f}^*)$, which concludes the proof.
	
	\textbf{Part 1}
	In this part, we want to show that $w(\mathbf{f}^G)+\epsilon \geq w(\mathbf{f}^*)$, which means the weight of the MGG schedule by ignoring the interference caused by users in $\mathcal{C}$ is greater than the weight of MaxWeight schedule.
	The following lemmas illustrate the relationship between the MGG schedule and MaxWeight schedule, and these results will be used later.
	\begin{lemma}\label{lem:marginal}
		Consider the MaxWeight schedule $\mathbf{f}^*=(u^*_1,\cdots, u^*_{\Omega},0\cdots,0)$. For each $i=1,\cdots, \Omega$, the marginal gain $\Delta^{\mathbf{f}^*, i}_{u^*_i}$ is always non-negative.
	\end{lemma}
	\begin{proof}
		Please see APPENDIX \ref{APP:L5}.
	\end{proof}
	\begin{remark}
		Similar to the MGG schedule generated by Algorithm \ref{alg2}, the MaxWeight schedule adds a user only if the marginal gain is non-negative. The only difference is that the MGG schedule will give higher priority to users with longer queue-length, whereas the MaxWeight schedule may skip some users with long queue lengths and choose other users with large marginal gain.
	\end{remark}
	
	In the MaxWeight schedule, for each user $u\in\mathcal{A}\cup\mathcal{B}$, we use $t_1(u)$ to denote the mini-slot that user $u$ is scheduled. In the MGG schedule, for each user $u\in\mathcal{N}$ we define $t_2(u)$ to be the mini-slot that its marginal gain is evaluated (either schedule $u$ or skip $u$ in $t_2(u)^{th}$ mini-slot), if $u$ has never been considered as a candidate, $t_2(u)=K$. 
	
	
	\begin{lemma}\label{lem:opt_greedy}
		In the MaxWeight schedule, for each $b\in\mathcal{B}$, consider user $d$ which has the longest queue-length among all users in group $g(b)$ that are not scheduled in the MGG schedule. We have: $t_1(b)<t_2(d)$, i.e., $b$ is scheduled earlier in the MaxWeight schedule than the time that $d$ is skipped in the MGG schedule.
	\end{lemma}
	\begin{proof}
		Please see APPENDIX \ref{APP:L6}.
	\end{proof}
	
	Define $N_{\mathcal{B}}(t)$ and $N_{\mathcal{C}}(t)$ to be the number of users in $\mathcal{B}$ and $\mathcal{C}$ scheduled in the MaxWeight and MGG schedule from the first mini-slot to $t^{th}$ mini-slot. We have the following lemma:
	\begin{lemma}\label{lem:BsmallC}
		For each $b\in\mathcal{B}$, which is scheduled in $t_1(b)^{th}$ mini-slot, we have $N_{\mathcal{B}}(t_1(b))\leq N_{\mathcal{C}}(t_1(b))$.
	\end{lemma}
	\begin{proof}
		Please see APPENDIX \ref{APP:L7}.
	\end{proof}
	
	From Lemma \ref{lem:BsmallC}, we can find a \textit{\textbf{mapping}} $h:\mathcal{B}\rightarrow\mathcal{C}$ , $i^{th}$ user $b_i$ in $\mathcal{B}$  corresponds to $i^{th}$ user $c_i$ in $\mathcal{C}$, such that $c_i$ is always scheduled earlier than $b_i$, i.e., $t_1(b_i)\geq t_2(c_i)$. For each user $b_i$, consider user $d_i$ which has the longest queue-length among all users in group $g(b_i)$ that are not scheduled in the MGG schedule. Note that users from group $g(b_i)$ only belongs to $\mathcal{A}$ or $\mathcal{B}$, user $d_i$ has the longest queue-length among all users in $\mathcal{B} \cap \mathcal{G}_{g(b_i)}$, thus $Q_{d_i}\geq Q_{b_i}$. From Lemma \ref{lem:opt_greedy}, we know $t_1(b_i)< t_2(d_i)$ and thus $t_2(c_i)<t_2(d_i)$. Then $Q_{c_i}\geq Q_{d_i}$ due to the LQF order of evaluating users in the MGG policy. Therefore, $Q_{c_i}\geq Q_{b_i}$.
	
	\begin{lemma}\label{lem:more}
		The MGG schedule will schedule more users than the MaxWeight schedule, i.e., $|\mathcal{B}|\leq |\mathcal{C}|$.
	\end{lemma}
	
	\begin{proof}
		Please see APPENDIX \ref{APP:L8}.
	\end{proof}

	Now we are ready to prove the result of part 1. Compare $w(\mathbf{f}^G)+\epsilon$ with $w(\mathbf{f}^*)$, we have two kinds of losses.
	
	\textbf{$\mathcal{A}$ loss:} For each user $a\in\mathcal{A}$, $a$ will be scheduled no earlier in the MGG schedule than that in the MaxWeight schedule, i.e., $t_1(a)\leq t_2(a)$ (corollary of Lemma \ref{lem:BsmallC}). Each user $a$ in the MGG schedule will receive lower or equal rate than that in the MaxWeight schedule. 
	
	\textbf{$\mathcal{B}$ loss:} In the MGG schedule, there is no weight contributed by users in $\mathcal{B}$. 
	
	If the total weight of the users in $\mathcal{C}$ can be used to cover $\mathcal{A}$ and $\mathcal{B}$ losses, then 
	$w(\mathbf{f}^G)+\epsilon\geq w(\mathbf{f}^*)$ holds. First, we consider $\mathcal{A}$ loss: let ${Loss}_{a_i}$ denote the weight loss on user $a_i$.
	\ifreport
	\begin{align}
	{Loss}_{a_i}&=Q_{a_i}(K-t_1(a_i)-|\{a\in\mathcal{A}|a \text{ is scheduled after } a_i \text{ in the MaxWeight schedule}\}|)\nonumber\\
	&\quad - Q_{a_i}(K-t_2(a_i)-|\{a\in\mathcal{A}|a \text{ is scheduled after } a_i \text{ in the MGG schedule}\}|)\nonumber\\
	&=Q_{a_i}(t_2(a_i)-t_1(a_i))\geq 0.
	\end{align}
	\fi
	
	Similarly, we use ${Loss}_{b_i}$ to denote the weight loss on user $b_i$:
	\begin{align}
	{Loss}_{b_i}=Q_{b_i}(K-t_1(b_i))\geq 0.
	\end{align}
	
	The weight difference $w(\mathbf{f}^{G})+\epsilon-w(\mathbf{f}^*)$ is the total weight of $\mathcal{C}$ minus $\mathcal{A}$ loss and $\mathcal{B}$ loss:
	\ifreport
	\begin{align}
	&w(\mathbf{f}^{G})+\epsilon-w(\mathbf{f}^*)\nonumber\\
	&=\sum_{i=1}^{|C|}Q_{c_i}(K-t_2(c_i))-\sum_{i=1}^{|\mathcal{A}|}{Loss}_{a_i}-\sum_{i=1}^{|\mathcal{B}|}{Loss}_{b_i}\nonumber\\
	&=\sum_{i=1}^{|C|}Q_{c_i}(K-t_2(c_i))-\sum_{i=1}^{|\mathcal{A}|}Q_{a_i}(t_2(a_i)-t_1(a_i))-\sum_{i=1}^{|\mathcal{B}|}Q_{b_i}(K-t_1(b_i)).\nonumber\\
	&\stackrel{(d)}\geq\sum_{i=1}^{|\mathcal{B}|}Q_{c_i}(t_1(b_i)-t_2(c_i))+\sum_{i=|\mathcal{B}|+1}^{|\mathcal{C}|}Q_{c_i}(K-t_2(c_i))-\sum_{i=1}^{|\mathcal{A}|}Q_{a_i}(t_2(a_i)-t_1(a_i))\nonumber\\
	&\stackrel{(e)}=\sum_{i=1}^{|\mathcal{C}|}Q_{c_i}(t_1(b_i)-t_2(c_i))-\sum_{i=1}^{|\mathcal{A}|}Q_{a_i}(t_2(a_i)-t_1(a_i)).\label{equ:g+E_opt}
	\end{align}
	\fi
	where inequality (d) comes from the property of mapping $h$ and equation (e) is derived by setting $t_1(b_i)=K$ for any dummy user $b_i$, $|\mathcal{B}|<i\leq |\mathcal{C}|$. Note that for each $i$, $t_1(b_i)-t_2(c_i)\geq 0$ and $t_2(a_i)-t_1(a_i)\geq 0$.
	
	\begin{lemma}\label{lem:nonnegative}
		The R. H. S. of (\ref{equ:g+E_opt}) is non-negative. 
	\end{lemma}
	\begin{proof}
		Please see APPENDIX \ref{APP_12}.
	\end{proof} 
	
	The result of Lemma \ref{lem:nonnegative} concludes the proof of part 1.
	
	\textbf{Part 2} In this part, we want to show that $\epsilon\leq 1/2w(\mathbf{f}^G)$, i.e., the extra weight is upper bounded by one half of the weight of the MGG schedule. We use $\epsilon_i$ and $w_i(\mathbf{f}^G)$ to denote the extra and actual weight from group $i$. It suffices to show a stronger (per-group) claim: For each group $i$, we have $\epsilon_i\leq 1/2w_i(\mathbf{f}^G)$.
	
	For each group $i$, note that we only need to consider the worst case where all the users from group $i$ are in $\mathcal{C}$. Otherwise, assume there are some users in $\mathcal{A}$, then $w_i(\mathbf{f}^G)$ remains the same while $\epsilon_i$ is smaller.
	
	\begin{lemma}
		Assume in the MGG schedule, we have $m$ users ($u_1, \cdots, u_m$, with queue-length $Q_{u_1}\geq \cdots \geq Q_{u_m}$) from group $i$, define $T_m$ to be the smallest rate of the last scheduled user such that the MGG schedule is feasible (marginal gain is always non-negative). Consider the case $K=K_m\triangleq T_m+t_2(u_m)$, we have $\epsilon^{K_m}_i\leq 1/2 w_i\left(\mathbf{f}^G_{K_m}\right)$, where $\epsilon^{K_m}_i$ and $w_i\left(\mathbf{f}^G_{K_m}\right)$ are extra weight and actual weight of $\mathbf{f}^G$ from group $i$ under $K_m$. \label{lem:induction}
	\end{lemma}
	\begin{proof}
		Please see APPENDIX \ref{APP:L9}.
	\end{proof}
	
	Note that $K_m$ is the smallest value of $K$ such that the MGG schedule is feasible, for any $K\geq K_m$, extra weight $\epsilon_i$ will be the same since it is only related to $u_1, \cdots, u_m$, however, $w_i(\mathbf{f}^G)$ will increase with $K$.
	\begin{align}
	\frac{\epsilon^K_i}{w_i(\mathbf{f}^G_{K})}\leq \frac{\epsilon^{K_m}_i}{w_i(\mathbf{f}^G_{K_m})}\le 1/2.
	\end{align}
	
	Therefore, we know for every feasible MGG schedule, $\epsilon_i/w_i(\mathbf{f}^G)$ is less than one half for any group $i=1,\cdots, I$. We finish the proof of part 2 and now we are able to show $w(\mathbf{f}^G)\geq 2/3w(\mathbf{f}^*)$.
\end{proof}

\begin{proposition}
	The 2/3 worst-case lower bound is tight in terms of weight.
\end{proposition}
\begin{proof}
	Assume $K=2^{r}$ for some positive integer $r>0$. All the users have the same queue-length, and there are $K-1$ groups where each group has enough users. Then the MaxWeight schedule will serve $K-1$ users, one for each group, which gives a total rate of $K(K-1)/2$, while the MGG Algorithm serves $K/2$ users from group 1, $K/4$ users from group 2, $\cdots$ and 1 user from group $r$, which gives a total rate of $(K^2-1)/3$. As $K\to\infty$, the efficiency ratio becomes arbitrarily close to $2/3$.
\end{proof}

\begin{theorem}\label{thm:th3}
	The throughput region of the proposed MGG policy is no smaller than the optimal throughput region under Half-duplex.
\end{theorem}
\begin{proof}
	We first prove the following lemma, which shows that the weight of MGG policy dominates the weight of any Half-duplex policy. 
	\begin{lemma}\label{lem:weight_domin}
		The weight of the MGG policy is no smaller than the maximum weight under Half-duplex, i.e., $w(\mathbf{f}^G)\geq w^*_{HD}$, where $w_{HD}(\cdot)$ is the total weight calculated under Half-duplex.
	\end{lemma}
	\begin{proof}
		Please see APPENDIX \ref{APP_L10}.
	\end{proof}
	
	Now we need to show that the MGG policy stabilizes any arrival vector $\boldsymbol{\lambda}=(\lambda_1, \cdots, \lambda_n)$ within the optimal throughput region under Half-duplex $\Lambda^*_{HD}$. The following lemma can be used to prove this claim. 
	\begin{lemma}\label{lem:lyapunov}
		Consider the capacity region $\Lambda_{HD}$ under Half-duplex, $w^*_{HD}$ is the maximum weight among all feasible scheduling policies under Half-duplex. If there exists a Full-duplex scheduling policy $\mathbf{f}^G$, such that $w(\mathbf{f}^G)\geq w^*_{HD}(\mathbf{f})$ for any queue-length vector, then policy $\mathbf{f}^G$ can stabilize any arrival vector within $\Lambda^*_{HD}$.
	\end{lemma}
	\begin{proof}
		Please see APPENDIX \ref{APP_11}.
	\end{proof}
	
Applying Lemma \ref{lem:weight_domin} and \ref{lem:lyapunov}, Theorem \ref{thm:th3} follows.
\end{proof}
\begin{remark}
Other promising low-complexity algorithms, such as greedily select users with the largest marginal gain or simply adopt certain amount of users from each group cannot work well either in the comparison with traditional Half-duplex schemes or under heterogeneous traffic arrivals.
\end{remark}

\section{Capacity Gain of Full-duplex over Half-duplex}\label{sec:discuss}
\ifreport
In this section, we will discuss the capacity gain of Full-duplex over Half-duplex. Let $\Lambda_{FD}$ and $\Lambda_{HD}$ denote the capacity region under Full-duplex and Half-duplex mode, respectively. To simplify, we only evaluate the capacity magnitude $\nu_{FD}$ and $\nu_{HD}$ along the $(1,\cdots, 1)$ vector (e.g., $(\nu_{FD},\cdots, \nu_{FD})$ is the largest arrival vector such that all users have the same arrival rate and the queuing system can be stabilized under Full-duplex mode). In addition, we assume all groups have the same size, i.e., $N_1=\cdots=N_I=N/I$.
\fi

For half-duplex, if the sum-rate is upper bounded by $B_{HD}$, then the lowest service rate is upper bounded by $B_{HD}/N$. According to the basic queuing theory, $\nu_{HD}\le B_{HD}/N$.
The sum-rate is calculated by:
\begin{align}
\sum_{i=1}^NR_i^{HD}=\left(K-\sum_{j=1}^Im_j\right)\sum_{j=1}^Im_j.\label{equ:HD:sum_rate}
\end{align}
where $m_j$ is the $j^{th}$ element in the user-selection vector. If $N\ge K/2$, the maximum of the sum-rate is achieved by taking $\sum_{j=1}^I m_j=K/2$, thus the upper bound $B_{HD}=\frac{K^2}{4}$. Otherwise, if $K$ is larger, the maximum is achieved by scheduling all users in the system, $B_{HD}=\left(K-N\right)N$. To sum up, 
\begin{align}
\nu_{HD}=
\begin{cases}
\frac{K^2}{4N}, &N\ge K/2\cr K-N, &\text{otherwise.}\end{cases}
\end{align}

Next, we will look at the Full-duplex case, consider a randomized policy $P$ which uses random schedules from time-slot to time-slot, denote its sum-rate as $B_{FD}$. Since the optimal throughput region is the union of the throuutghput regions of all possible scheduling policies, we have $\nu_{FD}\ge B_{FD}/N$. The sum-rate under $\mathbf{f}$ is calculated by:
\begin{align}
\sum_{i=1}^N R_i^{\mathbf{f}}=\sum_{j=1}^I\sum_{k<j}m_jm_k+\left(K-\sum_{j=1}^Im_j\right)\sum_{j=1}^Im_j.\label{equ:sum_rate}
\end{align}
where $m_j$ is the $j^{th}$ element in the user-selection vector $\mathbf{m}$.

The first term of the R. H. S. of (\ref{equ:sum_rate}) calculates the total rate from the first mini-slot to $\sum_{j=1}^Im_j^{th}$ mini-slot, we only need to count the number of user pairs $(u_i,u_j)$ such that $g(u_i)\neq g(u_j)$ and $u_i$ is scheduled before $u_j$. After $\sum_{j=1}^Im_j^{th}$ mini-slot, all scheduled user will have $K-\sum_{j=1}^Im_j$ additional rate. The total rate from the remaining mini-slot is just $\left(K-\sum_{j=1}^Im_j\right)\sum_{j=1}^Im_j$. To get the upper bound of the sum-rate, we need to solve the following maximization problem.
\begin{equation*}
\begin{aligned}
& \underset{m}{\text{maximize}}
& & \sum_{j=1}^I\sum_{k<j}m_jm_k+\left(K-\sum_{j=1}^Im_j\right)\sum_{j=1}^Im_j \\
& \text{subject to}
& & m_i\le N/I, m_i\in\mathbb{N}, \\
& & &\text{ for all } i=1,2, \cdots, I.
\end{aligned}
\end{equation*} 
If $N/I\ge \frac{K}{I+1}$ for all $i=1,2,\cdots, I$, then the maximum is achieved by taking $m_i=\frac{K}{I+1}$ for all $i=1,2,\cdots, I$. In this case, $B_{FD}=\frac{IK^2}{2(I+1)}$. Otherwise, the maximum is achieved by taking $m_i=N/I$ for all $i$. $B_{FD}=\frac{N(2IK-N-IN)}{2I}$. In a word,
\begin{align}
\nu_{FD}=
\begin{cases}
\frac{IK^2}{2N(I+1)}, &N\ge \frac{IK}{I+1}\cr \frac{2IK-N-IN}{2I}, &\text{otherwise.}\end{cases}.
\end{align}
Define Full-duplex gain $G_{FD}=\frac{\nu_{FD}}{\nu_{HD}}$, $\alpha=K/N$. We have:
\begin{align}
G_{FD}=
\begin{cases}
\frac{2I}{I+1}, &\alpha\le \frac{I+1}{I}
\cr \frac{2(2I\alpha-1-I)}{I\alpha^2}, &\frac{I+1}{I}\le \alpha\le 2
\cr 1+\frac{I-1}{2I(\alpha-1)} ,&\alpha\ge 2\end{cases}.\label{equ:case:gain}
\end{align}
Fix group number $I=10$, Fig. \ref{fig:gain_fixI} shows the Full-duplex gain $G_{FD}$ for different $\alpha$.
\begin{figure}[htbp]
	\centering
	\includegraphics[width=3.2in]{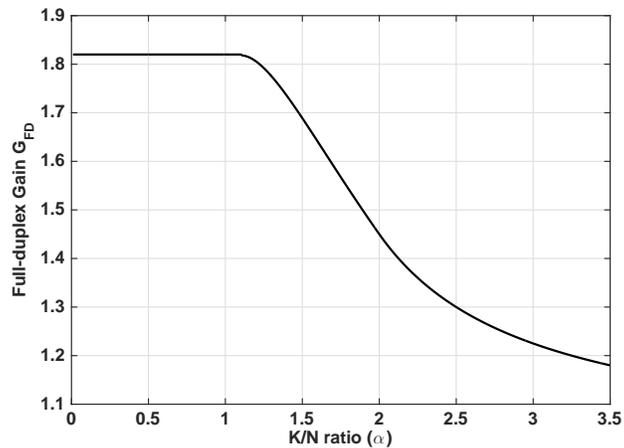}
	\vspace{-0.5cm}
	\caption{Full-duplex gain versus $\alpha$, when the group number $I=10$.}\label{fig:gain_fixI}
\end{figure}
As we can see in the figure, if $\alpha$ is smaller than 1.1, Full-duplex gain $G_{FD}$ remains larger than 1.8. In this regime, the number of users $N$ is larger than (or comparable to) $K$, which means the learning phase takes as long as nearly $K/2$ mini-slots. Note that the Full-duplex gain comes from concurrent channel probing and data transmission, the longer learning phase takes, the larger $G_{FD}$ will be observed. On the other hand, when $\alpha$ becomes larger, $G_{FD}$ decreases from 1.82 to 1.18. This is because the learning phase is negligible compared to $K$, thus we don't have much gain compared to the traditional schemes. In general, when $I$ becomes larger, the upper bound of the $G_{FD}$ becomes closer to 2, which matches the expected potential of the Full-duplex gain.

Fix $\alpha$ to be 1.0, 1.5 and 3, Fig. \ref{fig:gain_fixK/N} shows how does the Full-duplex gain $G_{FD}$ change with different group number $I$.
\begin{figure}[htbp]
	\centering
	\includegraphics[width=3.2in]{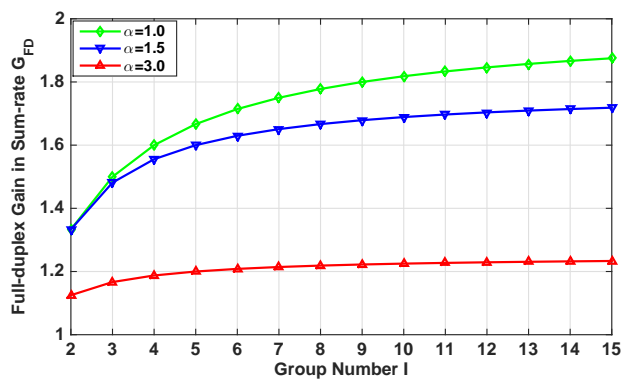}
	\caption{Full-duplex gain versus group number $I$, when the $K/N$ ratio ($\alpha$) is fixed.}
	\vspace{-0.0cm}\label{fig:gain_fixK/N}
\end{figure}
From Fig. \ref{fig:gain_fixK/N}, we can observe that the Full-duplex gain $G_{FD}$ keeps increasing as $I$ becomes larger. The scheduler has more flexibility when given more groups, thus a larger Full-duplex gain should be expected. Moreover, in many user regime (green and blue curve), $G_{FD}$ has improved by 40\% and 30\% when $I$ increases from 2 to 15. However, $G_{FD}$ does not improve much in small user regime (red curve). The learning phase only takes a small fraction of time, thus $G_{FD}$ is always a little larger than 1.1, regardless of what value $I$ takes.

\section{Numerical Results}\label{sec:sim}
In this section, we use simulations to evaluate our proposed greedy policy and compare
its performance with traditional Half-duplex and Full-duplex MaxWeght Scheduling (MWS) schemes.

\subsection{Simulation Settings}
We consider the downlink system of a single-cell Full-duplex MIMO system. There are $N$ users in this system and each user is equipped with only one antenna. The BS is assumed to have sufficiently large number of antennas. Suppose all users are divided into $I$ user groups such that users from different group does not interfere with each other. Unlike the assumption we make in Section \ref{sec:discuss}, each user group now could have different group size. In addition, we assume that each time-slot has 15 mini-slots, i.e., $K=15$. We consider i.i.d. arrival, i.e., 
\begin{align}
A_u[t]=\begin{cases}
K, &\text{w.p. } \lambda
\cr 0, &\text{otherwise}\end{cases}\nonumber
\end{align}
where $\lambda$ is the scaled arrival rate of queue $u$, $u\in\mathcal{N}$.

\subsection{Performance of Greedy Policy under Different Regimes}
Fix group number $I=4$, we then evaluate the performance of the proposed greedy policy in three regimes which represent three conditions of (\ref{equ:case:gain}). Define regime 1 as the many-user regime such that $\alpha\leq 1.25$. In regime 1, we take $N_1=8, N_2=5, N_3=6, N_4=1$, with sum $N=20$ and $\alpha=0.75$.
Regime 2 denotes the moderate regime, where $N$ is comparable with $K$ such that $1.25\le \alpha\le 2$. In regime 2, $N_1=3, N_2=2, N_3=2, N_4=3$, with sum $N=10$ and $\alpha=1.5$. Regime 3 represents the small-user regime such that $\alpha\ge 2$. 
In regime 3, we take $N_1=1, N_2=1, N_3=1, N_4=1$, with sum $N=4$ and $\alpha=3.75$. For all these three scenarios, we plot the average queue-length under different arrival rate $\lambda$ in Fig. \ref{fig:sim_queue}.

\ifreport
\begin{figure*}[t]
	\centering \subfigure[Regime 1]{
		\includegraphics[width=1.9in]{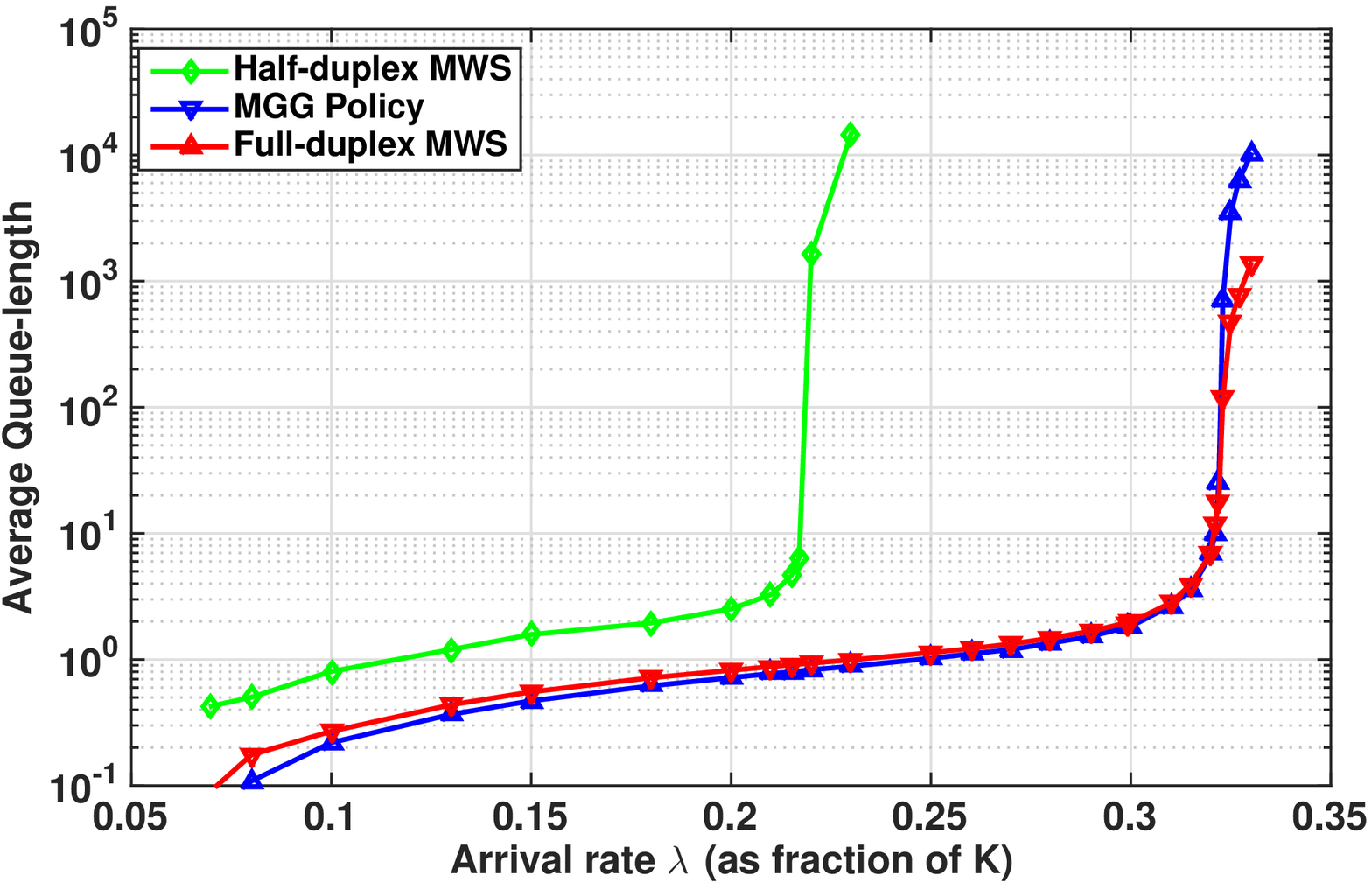}}
	\hspace{0.2in}
	\subfigure[Regime 2]{
		\includegraphics[width=1.9in]{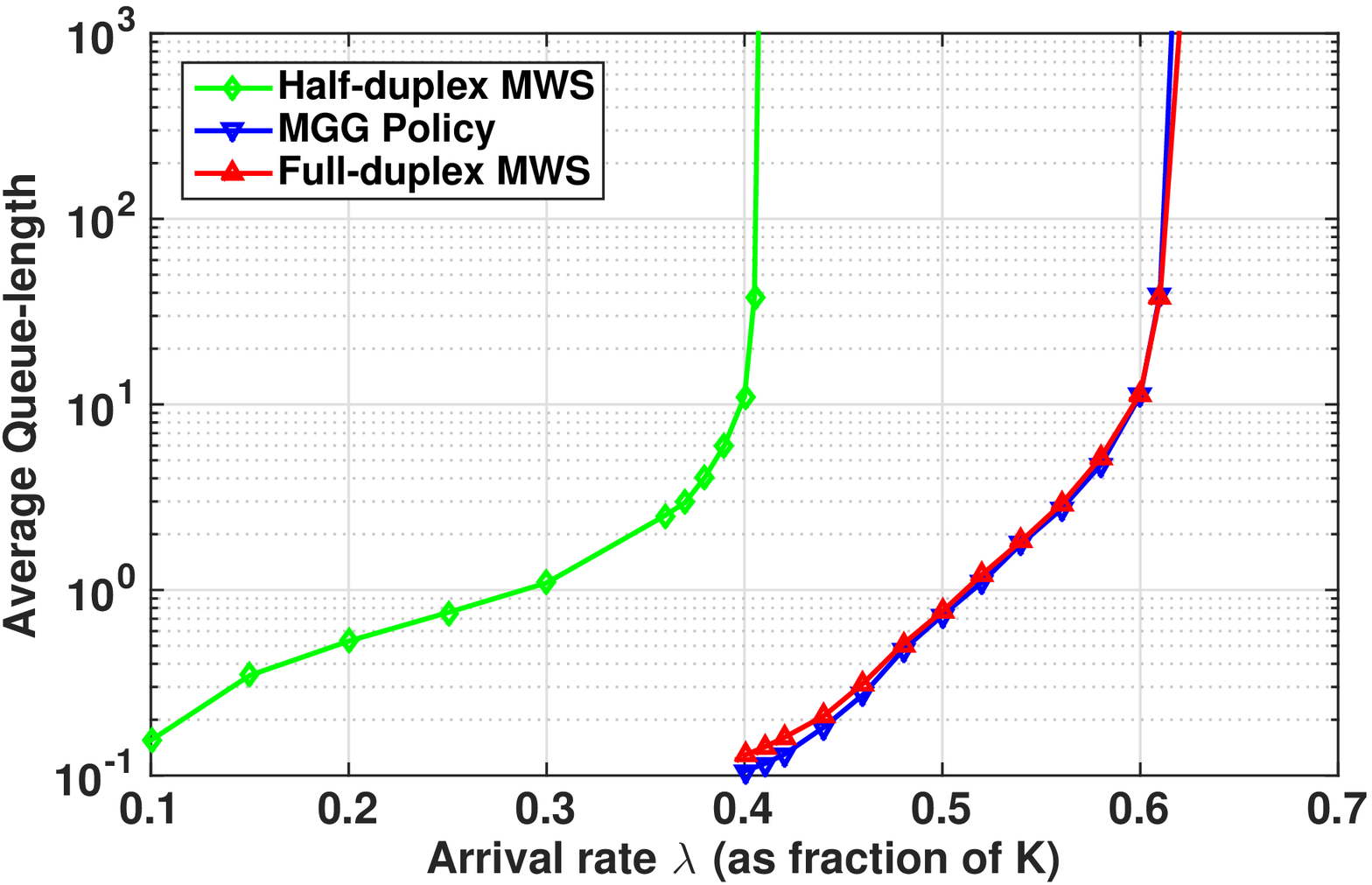}}
	\hspace{0.2in}
	\subfigure[Regime 3]{
		\includegraphics[width=1.9in]{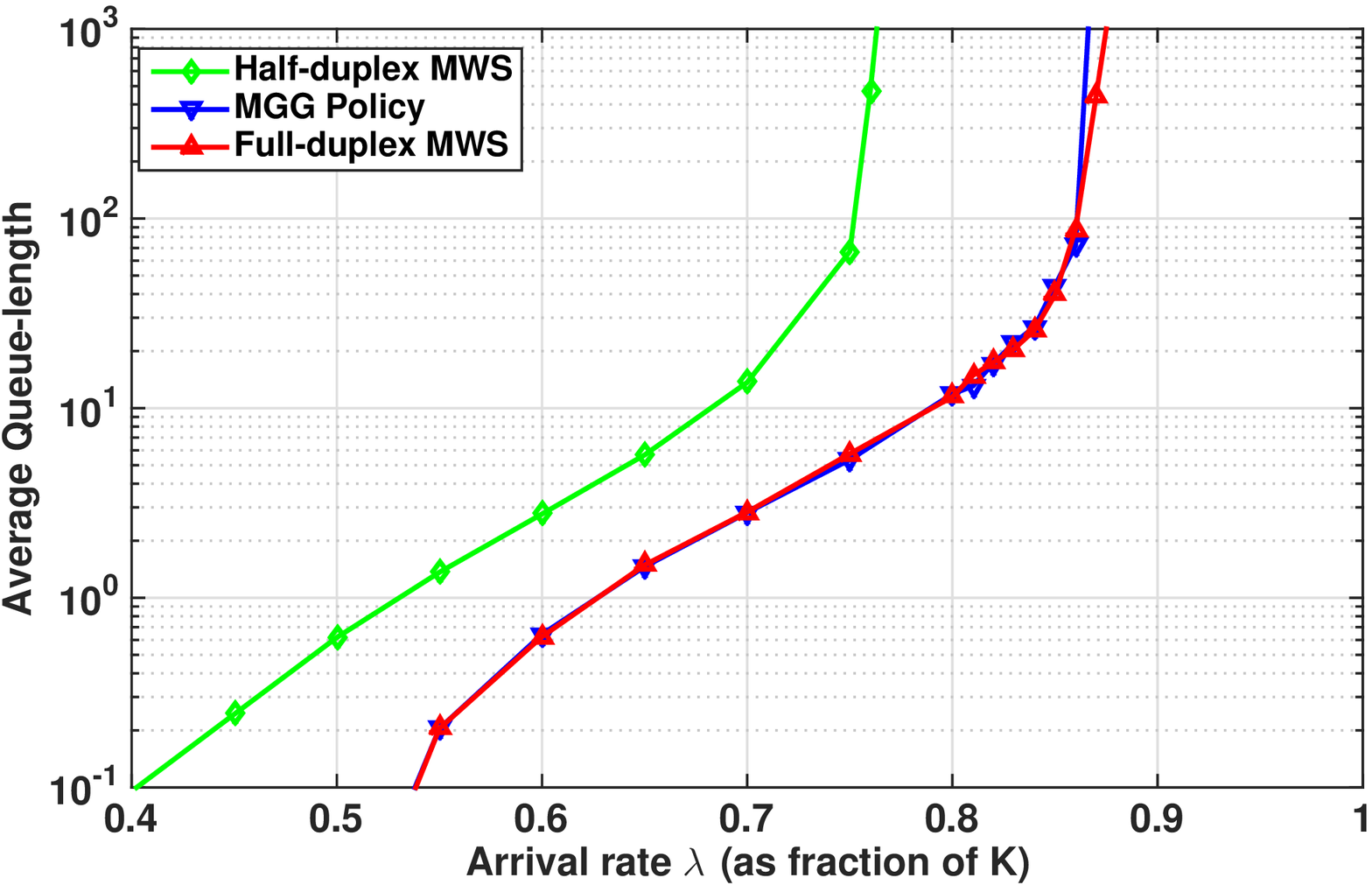}}
	\vspace{-0.3cm}
	\caption{Average queue-length under different arrival rate.}
	\vspace{-0.5cm}
	\label{fig:sim_queue}
\end{figure*}
\else
\fi
In all three regimes, the performance of the MGG policy is very close to the Full-duplex MaxWeight policy. Thus, the throughput performance of the MGG policy is also very close to optimal. 
The Full-duplex gain is larger if $\alpha$ is small, meaning $K$ is smaller compared to $N$. In this case, the control overhead of sending probing signals becomes the system bottleneck. Introducing Full-duplex reduces the control overhead and thus the throughput is improved substantially. As $\alpha$ becomes larger, the control overhead no longer limits the throughput, since it only takes a small fraction of time to send probing signals. 
As a result, Full-duplex gain decreases from 1.5 to 1.13 from as $\alpha$ increases from 0.75 to 3.75.

\subsection{Performance of Greedy Policy under Random Group Assignments}
Given $N$ users, the way of assigning users to different groups affects the Full-duplex gain. In this section, we would like to evaluate throughput performance under random group assignments. Fix group number $I=4$, number of users $N=10$ and $K=15$. Assume that each user has equal probability to be assigned to each group, the following figure shows the empirical CDF of the Full-duplex gain for 10000 samples of random group assignments.
\begin{figure}[htbp]
	\centering
	\includegraphics[width=3.2in]{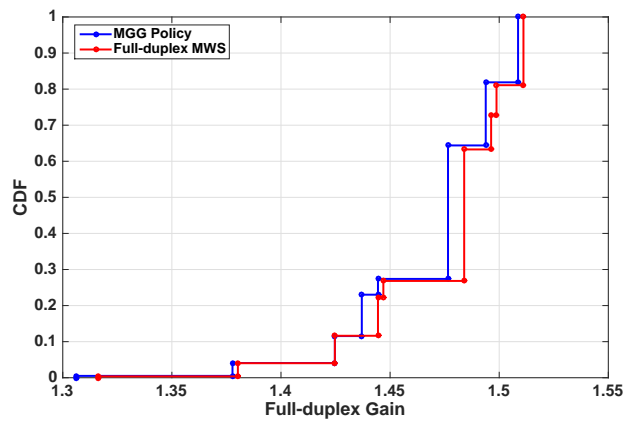}
	\vspace{-0.3cm}
	\caption{The empirical CDF for Full-duplex gain compared to Half-duplex throughput optimal policy}
	\vspace{-0.1cm}\label{fig:cdf}
\end{figure}

From Fig. \ref{fig:cdf}, we can observe that the Full-duplex gain of the MGG policy and MaxWeight policy have similar distributions. Although in theory there may exist scenarios in which the MGG policy is suboptimal, in typical scenarios it achieves near-optimal throughput performance. The median Full-duplex gain under the MaxWeight scheduling and the MGG policy is around 1.48. Although the lowest Full-duplex gain is around 1.3, in typical scenarios (90\% of all samples), the Full-duplex gain is larger than 1.44 (44\% improvement).

\section{Conclusion}\label{sec:conclusion}
In this paper, we develop a throughput optimal scheduling policy for concurrent channel probing and data transmission scheme. To further reduce the complexity when there are a large number of groups, we propose a greedy policy with complexity $O(N\log N)$ that not only achieves at least 2/3 of the optimal throughput region but also outperforms any feasible Half-duplex solutions. Furthermore, we derive the Full-duplex gain under different system parameters. Finally, we use numerical simulations to validate our theoretical results. 

\ifreport

{\scriptsize
\bibliographystyle{ieeetr}
\bibliography{reference}  
}
\fi

\begin{appendices}
\section{Proof of Lemma \ref{lem:shift_zero}}\label{APP:L1}
Assume we have a scheduling vector $\mathbf{f}=(u_1,\cdots, u_i, 0, u_{i+2}, \cdots, u_K)$ and the shifted version $\mathbf{f'}=(u_1, \cdots, u_i, u_{i+2}, \cdots, u_K, 0)$. We have:
\begin{align}
w(\mathbf{f'})-w(\mathbf{f})&=\sum_{j=1}^K Q_{u_j} R_{u_j}^{\mathbf{f'}}-\sum_{j=1}^K Q_{u_j} R_{u_j}^{\mathbf{f}}\nonumber\\
&=\left(\sum_{j=1}^ i Q_{u_j} R_{u_j}^{\mathbf{f'}}+\sum_{j=i+2}^ K Q_{u_j} R_{u_j}^{\mathbf{f'}}\right)-\left(\sum_{j=1}^ i Q_{u_j} R_{u_j}^{\mathbf{f}}+\sum_{j=i+2}^ K Q_{u_j} R_{u_j}^{\mathbf{f}}\right).\label{weight_diff}
\end{align}
Note that for any $j\leq i$, we have $R_{u_j}^{\mathbf{f'}}=R_{u_j}^{\mathbf{f}}$ and
\begin{align}
R_{u_j}^{\mathbf{f}}=
\begin{cases}
\sum\limits_{t=j+1}^i\mathbf{1}_{\{g(u_j)\neq g(u_t)\}}+\sum\limits_{t=i+2}^K\mathbf{1}_{\{g(u_j)\neq g(u_t)\}}, &j<i\cr \sum\limits_{t=i+2}^K\mathbf{1}_{\{g(u_j)\neq g(u_t)\}}, &j=i\end{cases}.\label{rate_j_i}
\end{align}
For any $j\geq i+2$, we have:
\begin{align}
R_{u_j}^{\mathbf{f'}}=\sum_{t=j+1}^K \mathbf{1}_{\{g(u_j)\neq g(u_t)\}}+1=R_{u_j}^{\mathbf{f}}+1.\label{rate_i_j}
\end{align}
Substituting (\ref{rate_j_i}) and (\ref{rate_i_j}) into (\ref{weight_diff}), we have:
\begin{align}
w(\mathbf{f'})-w(\mathbf{f})=\sum_{j=i+2}^KQ_{u_j}\geq 0.
\end{align}

\section{Proof of Lemma \ref{lem:LQF}}\label{APP:L2}
We will prove the Lemma by showing that for fixed $K$ and any fixed user set of $\Omega$ users $\{u_1,\ldots,u_\Omega\}$, among all scheduling vectors LQF 
-type of scheduling vector maximizes the weight $w(\mathbf{f})$. We use mathematical induction to prove this claim.

\textbf{\emph{Base case:}} $\Omega=2$, i.e., there are only two users in the network, user $u_1$ and user $u_2$. Without loss of generality, we assume $Q_{u_1}\geq Q_{u_2}$. According to Lemma \ref{lem:shift_zero}, we only need to consider two strategies $\mathbf{f_1}=(u_1, u_2,0, \cdots, 0)$ and $\mathbf{f_2}=(u_2, u_1, 0, \cdots, 0)$. Note that $\mathbf{f_1}$ is the scheduling vector generated by LQF.

\emph{Case 1:} user $u_1$ and user $u_2$ are from different groups, i.e., $g(u_1)\neq g(u_2)$. 

We can compute $w(\mathbf{f_1})$ and $w(\mathbf{f_2})$ as follows:
\begin{align}
w(\mathbf{f_1})=(K-1)Q_{u_1}+(K-2)Q_{u_2}.
\end{align}
\begin{align}
w(\mathbf{f_2})=(K-1)Q_{u_2}+(K-2)Q_{u_1}.
\end{align}

Since $Q_{u_1}\geq Q_{u_2}$, we have $w(\mathbf{f_1})\geq w(\mathbf{f_2})$.

\emph{Case 2:} user $u_1$ and user $u_2$ come from the same group, i.e., $g(u_1)=g(u_2)$. 

Both strategies $\mathbf{f_1}$ and $\mathbf{f_2}$ end up with the same weight $(K-2)(Q_{u_1}+Q_{u_2})$.

Combining both cases, we can show that $w(\mathbf{f_1})\geq w(\mathbf{f_2})$, which implies scheduling vector $\mathbf{f_1}$ maximizes the total weight $w(\mathbf{f})$.

\textbf{\emph{Inductive hypothesis:}} Assume the LQF-type of scheduling vector maximizes the weight $w(\mathbf{f})$ for any fixed scheduled user set with $\Omega-1$ users.

\textbf{\emph{Inductive step:}} We need to show that the scheduling vector generated by LQF maximizes weight $w(\mathbf{f})$ for any fixed scheduled user set $\{u_1, \cdots, u_{\Omega}\}$. Without loss of generality, we assume $Q_{u_1}\geq\cdots\geq Q_{u_{\Omega}}$. Let $(u_{\sigma_1},\cdots, u_{\sigma_{\Omega}})$ denote an arbitrary permutation of $(u_1,\cdots, u_{\Omega})$, and its resulting scheduling vector $\mathbf{f^{\sigma}}=(u_{\sigma_1}, \cdots, u_{\sigma_{\Omega}}, 0, \cdots, 0)$. 
Let $\mathbf{f_1}=(u_1, \cdots, u_{\Omega}, 0, \cdots, 0)$ denote the scheduling vector generated by LQF.
It is equivalent to show $w(\mathbf{f}_1)\geq w(\mathbf{f^{\sigma}})$.

Start from an arbitrary scheduling vector $\mathbf{f^{\sigma}}=(u_{\sigma_1},\cdots, u_{\sigma_{\Omega}}, 0, \cdots, 0)$.
Let's fix the first element in $\mathbf{f^{\sigma}}$, say $u_{\sigma_1}=u_c$. What is the resulting optimal schedule? Recall that our goal is to find a permutation $\{u_{\sigma_2},\cdots, u_{\sigma_{\Omega}}\}$ of $\{u_1,\cdots, u_{\Omega}\}\setminus \{u_c\}$, such that $\sum_{i=1}^{\Omega}Q_{u_{\sigma_i}}R_{u_{\sigma_i}}^{\mathbf{f}}$ is maximized.
\begin{align}
&\quad\max_{u_{\sigma_2}, \cdots, u_{\sigma_{\Omega}}}\sum_{i=1}^\Omega Q_{u_{\sigma_i}}R_{u_{\sigma_i}}^{\mathbf{f}}\nonumber\\
&=\max_{u_{\sigma_2}, \cdots, u_{\sigma_{\Omega}}}\sum_{i=1}^\Omega Q_{u_{\sigma_i}}\sum_{j=i+1}^{K}\mathbf{1}_{\{g(u_{\sigma_i})\neq g(u_{\sigma_j})\}}\nonumber\\
&=\max_{u_{\sigma_2}, \cdots, u_{\sigma_{\Omega}}}\left(Q_{u_c}\sum_{j=2}^K\mathbf{1}_{\{g(u_c)\neq g(u_{\sigma_j})\}}+\sum_{i=2}^\Omega Q_{u_{\sigma_i}}\sum_{j=i+1}^{K}\mathbf{1}_{\{g(u_{\sigma_i})\neq g(u_{\sigma_j})\}}\right).\label{induction_objective}
\end{align}

The first term of R. H. S. of (\ref{induction_objective}) is the same for any schedules (permutations), hence we only need to focus on the second term. The optimal value of $u_{\sigma_2},\cdots, u_{\sigma_{\Omega}}$ solves the following optimization problem (P1):
\begin{align}
\max_{u_{\sigma_2}, \cdots, u_{\sigma_{\Omega}}}\sum_{i=2}^\Omega Q_{u_{\sigma_i}}\sum_{j=i+1}^{K}\mathbf{1}_{\{g(u_{\sigma_i})\neq g(u_{\sigma_j})\}}.
\end{align}

Let ${u}_{\sigma_j}'=u_{\sigma_{j+1}}$ for any $j=1,2, \cdots, \Omega-1$, we can rewrite P1 as:
\begin{align}
\max_{{u}_{\sigma_1}', \cdots, {u}_{\sigma_{\Omega-1}}'}\sum_{i=1}^{\Omega-1}Q_{{u}_{\sigma_i}'}\sum_{j=i+1}^{K-1} \mathbf{1}_{\{g({u}_{\sigma_i}')\neq g({u}_{\sigma_j}')\}}.
\end{align}
The optimal scheduling vector ${\mathbf{f}'}^*$ of P1 satisfies:
\begin{align}
{\mathbf{f}'}^*&=\argmax\sum_{i=1}^{\Omega-1}Q_{{u}_{\sigma_i}'}\sum_{j=i+1}^{K-1} \mathbf{1}_{\{g({u}_{\sigma_i}')\neq g({u}_{\sigma_j}')\}}\nonumber\\
&\stackrel{(a)}{=}\argmax\left(\sum_{i=1}^{\Omega-1}Q_{{u}_{\sigma_i}'}\sum_{j=i+1}^{K-1} \mathbf{1}_{\{g({u}_{\sigma_i}')\neq g({u'}_{\sigma_j})\}}\qquad+\sum_{u\subset\{u_1, \cdots, u_{\Omega}\}\setminus\{u_c\}}Q_u\right)\nonumber\\
&\stackrel{(b)}{=}\argmax\sum_{i=1}^{\Omega-1}Q_{{u}_{\sigma_i}'}\ \left(\sum_{j=i+1}^{K-1} \mathbf{1}_{\{g({{u}_{\sigma_i}'})\neq g({u'}_{\sigma_j})\}}+1\right).
\label{equ:P1_P2}
\end{align}
(a) holds since $\sum_{u\subset\{u_1, \cdots, u_{\Omega}\}\setminus\{u_c\}}Q_u$ is a constant for any possible schedule using users $\{u_1, \cdots, u_{\Omega}\}$. (b) holds since $u_{\sigma_1}',\cdots, u_{\sigma_{\Omega-1}}'$ is a permutation of $\{u_1,\cdots, u_{\sigma_{\Omega}}\}\setminus\{u_c\}$.

Equation (\ref{equ:P1_P2}) implies that the optimal scheduling vector ${\mathbf{f}'}^*$ also solves the MaxWeight problem for scheduled set with  $\Omega-1$ users \{${u}_{\sigma_1}', \cdots, {u}_{\sigma_{\Omega-1}}'\}$, equivalently, $\{u_{\sigma_2}, \cdots, u_{\sigma_{\Omega}}\}$. From inductive hypothesis, we know:
\begin{align}
{\mathbf{f}'}^*=
\begin{cases}
(u_1,\cdots, u_{c-1}, u_{c+1}, \cdots, u_{\Omega}, 0, \cdots, 0), &c\neq 1\cr (u_2,\cdots, u_{\Omega}, 0, \cdots, 0), &c=1\end{cases}.
\end{align}

Hence, after fixing the first element $u_{\sigma_1}=u_c$, the resulting optimal scheduling vector will take the following form:
\begin{align}
{\mathbf{f}_c}=
\begin{cases}
(u_c, u_1,\cdots, u_{c-1}, u_{c+1}, \cdots, u_{\Omega}, 0, \cdots, 0), &c\neq 1\cr (u_1, u_2,\cdots, u_{\Omega}, 0, \cdots, 0), &c=1\end{cases}.
\end{align}

Next, we only need to prove $w(\mathbf{f}_1)\geq w(\mathbf{f}_c)$ for any $c\neq 1$. 
Note that $\mathbf{f}_1$ and $\mathbf{f}_c$ agree on the scheduled users from $c+1^{th}$ mini-slot to the end, hence $R_{u_i}^{\mathbf{f}_1}=R_{u_i}^{\mathbf{f_c}}$ for any $i\in\{c+1, \cdots, \Omega\}$. The weight difference only comes from user set $\{u_1, \cdots, u_c\}$.
\begin{align}
&w(\mathbf{f}_1)-w(\mathbf{f}_c)\nonumber\\
&=\sum_{i=1}^{\Omega}Q_{u_i}\sum_{j=i+1}^{K}\mathbf{1}_{\{g(u_i)\neq g(u_j)\}}\nonumber\\
&\quad-\left(Q_{u_c}\sum_{j=1, j\neq c}^K\mathbf{1}_{\{g(u_c)\neq g(u_j)\}}+\sum_{i=1, i\neq c}^{\Omega}Q_{u_i}\sum_{j=i+1, j\neq c}^{K}\mathbf{1}_{\{g(u_i)\neq g(u_j)\}}\right)\nonumber\\
&=\sum_{i=1}^{c-1}Q_{u_i}\mathbf{1}_{\{g(u_i)\neq g(u_c)\}}-Q_{u_c}\sum_{j=1}^{c-1}\mathbf{1}_{\{g(u_c)\neq g(u_j)\}}\nonumber\\
&=\sum_{i=1}^{c-1}\left(Q_{u_i}-Q_{u_c}\right)\mathbf{1}_{\{g(u_i)\neq g(u_c)\}}\geq 0.\label{weight_difference}
\end{align}
The last inequality of (\ref{weight_difference}) comes form the assumption $Q_{u_1}\geq Q_{u_2} \geq \cdots \geq Q_{u_{\Omega}}$.


\section{Proof of Lemma \ref{lem:chooselargest}}\label{APP:L3}
If the set $\mathcal{U}^{\mathbf{f}}_i/\ \mathcal{P}^{\mathbf{m}}_i$ is non-empty, that means there is a user $u_l$ whose queue-length is among the $m_i$ longest queues, but is not selected by $\mathbf{f}$. The schedule $\mathbf{f}$, instead, chooses another user $u_s$, whose queue-length is not among the $m_i$ longest queue in group $i$. Thus, we know $g(u_l)=g(u_s)=i$ and $Q_{u_l}\geq Q_{u_s}$. By replacing $u_s$ by $u_l$, $u_l$ may be scheduled earlier than $u_s$ due to its larger queue-length, which means some users with queue-lengths in between $Q_{u_s}$ and $Q_{u_l}$ may have to be scheduled one mini-slot later. This action will only affect those users that are not from group $i$, since users from the same group are not able to transmit any packet anyway.
Denote $\mathcal{\mathcal{Y}}$ to be the set of such users that its rate is affected by replacing $u_s$ with $u_l$.
\begin{align}
w(\mathbf{f}')-w(\mathbf{f})&=Q_{u_l}R^{\mathbf{f}'}_{u_l}-Q_{u_s}R^{\mathbf{f}}_{u_s}-\sum_{y\in \mathcal{Y}}Q_y\nonumber\\
&=Q_{u_l}(R^{\mathbf{f}'}_{u_l}-R^{\mathbf{f}}_{u_s})+(Q_{u_l}-Q_{u_s})R^{\mathbf{f}}_{u_s}+\sum_{y\in \mathcal{Y}}Q_y\nonumber\\
&\geq Q_{u_l}(R^{\mathbf{f}'}_{u_l}-R^{\mathbf{f}}_{u_s})-\sum_{y\in \mathcal{Y}}Q_y\nonumber\\
&\stackrel{(c)}=\sum_{y\in \mathcal{Y}}(Q_{u_l}-Q_y)\geq 0.
\end{align}
equality (c) holds since $|\mathcal{Y}|=R^{\mathbf{f}'}_{u_l}-R^{\mathbf{f}}_{u_s}$, which are equivalent ways to count the number of users whose rate are affected.
\section{Proof of Lemma \ref{lem:chooselargest2}}\label{APP:L4}
Suppose there exists a schedule $\mathbf{f}$ which maximizes weight $w(\mathbf{f})$ but it does not pick users with longest queue-length in a certain group $i$. In this case, $\mathcal{U}^{\mathbf{f}}_i/\ \mathcal{P}^{\mathbf{m}}_i$ is not empty and we can find $u_s\in \mathcal{U}^{\mathbf{f}}_i/\ \mathcal{P}^{\mathbf{m}}_i$ and $u_l\in \mathcal{P}^{\mathbf{m}}_i/\ \mathcal{U}^{\mathbf{f}}_i$, such that $Q_{u_s}<Q_{u_l}$. From Lemma \ref{lem:chooselargest}, we know by replacing $Q_{u_s}$ by $Q_{u_l}$, we can maintain the same $\mathbf{m}$ while the total weight can be strictly increased. This fact contradicts with the assumption that $\mathbf{f}$ maximizes the weight $w(\mathbf{f})$.
\section{Proof of Lemma \ref{lem:marginal}}\label{APP:L5}
Assume there exists $\tau_1<\cdots<\tau_{s}$, such that $\Delta^{\mathbf{f}, \tau_1}_{u^*_{\tau_1}}<0, \cdots, \Delta^{\mathbf{f}, \tau_s}_{u^*_{\tau_s}}<0$. Start with $\tau_s$, we can show that the weight will strictly increase by skipping this user $u_{\tau_s}$ and continue service with the remaining users in $\mathbf{f}^*$. Use $w(\mathbf{f}')$ to denote the weight of the schedule without user $u_{r_s}$.
\begin{align}
w(\mathbf{f}')-w(\mathbf{f}^*)&=\underbrace{\sum_{j=1}^{\tau_s-1}Q_{u^*_j}\mathbf{1}_{\{g(u^*_j)=g(u^*_{\tau_s})\}}}_{\text{No longer block early scheduled users}}+\underbrace{\sum_{j=\tau_s+1}^{\Omega}Q_{u^*_{j}}}_{\text{Late users scheduled one mini-slot earlier}}\nonumber\\
&\quad-\underbrace{Q_{u^*_{\tau_s}}\sum_{j=\tau_s+1}^K\mathbf{1}_{\{g(u^*_j)\neq g(u^*_{\tau_s})\}}}_{\text{No weight contributed by user $u^*_{\tau_s}$}}\nonumber\\
&=\sum_{j=1}^{\tau_s-1}Q_{u^*_j}\mathbf{1}_{\{g(u^*_j)=g(u^*_{\tau_s})\}}+\sum_{j=\tau_s+1}^{\Omega}Q_{u^*_{j}}-Q_{u^*_{\tau_s}}\left(K-\tau_s-\sum_{j=\tau_s+1}^K\mathbf{1}_{\{g(u^*_j)= g(u^*_{\tau_s})\}}\right)\nonumber\\
&\geq \sum_{j=1}^{\tau_s-1}Q_{u^*_j}\mathbf{1}_{\{g(u^*_j)=g(u^*_{\tau_s})\}}-Q_{u^*_{\tau_s}}\left(K-\tau_s\right)\nonumber\\
&=-\Delta^{\mathbf{f}, \tau_s}_{u^*_{\tau_s}}>0.
\end{align}

Continuing with the same procedure for $\tau_{s-1}, \cdots, \tau_1$, the total weight will keep increasing, then we can come up with a new schedule such that the marginal gain is always non-negative, and its weight is strictly larger than the MaxWeight schedule, which contradicts the optimality of the original schedule. Therefore, Lemma \ref{lem:marginal} holds.

\section{Proof of Lemma \ref{lem:opt_greedy}}\label{APP:L6}
We prove the lemma by contradiction, suppose there exists user $b_0$ and $d_0$, such that $t_1(b_0)\geq t_2(d_0)$ and $g(b_0)=g(d_0)$. From Lemma $\ref{lem:chooselargest2}$ we know, both MaxWeight and MGG schedule will always schedule users with longest queue-length in each group, but they may schedule different number of users in each group since they may have different user-selection vectors. If the MaxWeight schedule picks more users in group $i$, then there will be some users in group $i$ that is in $\mathcal{B}$ and no user in group $i$ will be in $\mathcal{C}$. Otherwise, there will be some users in group $i$ that is in $\mathcal{C}$ and no user in group $i$ will be in $\mathcal{B}$. Consider group $g(b_0)$, no user in group $g(b_0)$ is in $\mathcal{C}$ because the MaxWeight schedule picks more users in group $i$. Since user $d_0$ is skipped in $t_2(d_0)^{th}$ mini-slot, which means the marginal gain $\Delta^{\mathbf{f}^G, t_2(d_0)}_{d_0}<0$ (guaranteed by the MGG policy):
\begin{align}
\Delta^{\mathbf{f}^G, t_2(d_0)}_{d_0}&=Q_{d_0}\left(K-t_2(d_0)\right)-\underbrace{\sum_{i=1}^{t_2(d_0)-1}Q_{u_i}\mathbf{1}_{\{g(u_i)=g(d_0)\}}}_{\text{Sum of queue-lengths in $\mathcal{A}\cap \mathcal{G}_{g(d_0)}$}}\nonumber\\
&<0.
\end{align}
Consider the marginal gain of user $b_0$:
\begin{align}
\Delta^{\mathbf{f}^G, t_1(b_0)}_{b_0}&=Q_{b_0}\left(K-t_1(b_0)\right)-\underbrace{\sum_{i=1}^{t_1(b_0)-1}Q_{u_i}\mathbf{1}_{\{g(u_i)=g(b_0)\}}}_{\text{Sum of queue-lengths in $\mathcal{A}\cap \mathcal{G}_{g(b_0)}$}}\nonumber\\
&\leq Q_{b_0}(K-t_2(d_0))-\sum_{i=0}^{t_2(d_0)-1}Q_{u_i}\mathbf{1}_{\{g(u_i)=g(d_0)\}}\nonumber\\
&=\Delta^{\mathbf{f}^G, t_2(d_0)}_{d_0}<0.
\end{align}
which contradicts the result of Lemma \ref{lem:marginal}, hence Lemma \ref{lem:opt_greedy} holds.

\section{Proof of Lemma \ref{lem:BsmallC}}\label{APP:L7}
Note that for a given schedule, if we have several users with the same queue-length, the ordering of these users does not affect the weight. Therefore, given MaxWeight and MGG schedules, we can reorder the schedule such that users with the same queue-length will follow ``$\mathcal{A}$ first, $\mathcal{B}/\mathcal{C}$ second'' order. Given $b\in\mathcal{B}$, we can find user $d$, such that $d$ has the longest queue-length among all users in group $g(b)$ that are not scheduled in the MGG schedule. From Lemma \ref{lem:opt_greedy}, we have $t_1(b)<t_2(d)$. Furthermore, $Q_b\leq Q_d$ from the definition of user $d$. Let $N_{\mathcal{A}}^{*}(t)$ and $N_{\mathcal{A}}^{G}(t)$ denote the total number of users in $\mathcal{A}$ from 1 to $t$ mini-slot in the MaxWeight and MGG schedule. We have:
\begin{align}
N_{\mathcal{A}}^{G}(t_1(b))&\leq N_{\mathcal{A}}^{G}(t_2(d))\leq |\{a\in\mathcal{A}|Q_a\geq Q_d\}|\nonumber\\
&\leq |\{a\in\mathcal{A}|Q_a\geq Q_b\}|=N_{\mathcal{A}}^{*}(t_1(b)).\label{equ:N_A_greedy_opt}
\end{align}

By definition, $N_{\mathcal{B}}(t_1(b))=t_1(b)-N_{\mathcal{A}}^{*}(t_1(b))$ and $N_{\mathcal{C}}(t_1(b))=t_1(b)-N_{\mathcal{A}}^{G}(t_1(b))$. Combining with (\ref{equ:N_A_greedy_opt}), we have $N_{\mathcal{B}}(t_1(b))\leq N_{\mathcal{C}}(t_1(b))$.

\section{Proof of Lemma \ref{lem:more}}\label{APP:L8}
Let $b_{L}$ denote the last scheduled user in the MaxWeight schedule that belongs to $\mathcal{B}$. 

\textbf{Case 1:} $b_L$ is the last scheduled user in $\mathbf{f}^*$.

Consider user $d_L$, which has the longest queue-length
among all users in group $g(b_L)$ that are not scheduled in the MGG schedule, From Lemma \ref{lem:opt_greedy}, $t_2(d_L)>t_1(b_L)$. Note that $t_1(b_L)$ is the total number of users $\mathbf{f}^*$ schedules, $\mathbf{f}^G$ schedules at least $t_2(d_L)-1 (\geq t_1(b_L))$ users. 

\textbf{Case 2:} $b_L$ is not the last scheduled user in $\mathbf{f}^*$.

After $t_1(b_L)^{th}$ mini-slot, $\mathbf{f}^*$ schedules users in $\mathcal{A}$ only. Total number of users scheduled in $\mathbf{f}^*$ is $t_1(b_L)+|\mathcal{A}|-N^{*}_{\mathcal{A}}(t_1(b_L))$. From Lemma \ref{lem:BsmallC}, we know 
\begin{align}
N^G_{\mathcal{A}}(t_1(b_L))\leq N^{*}_{\mathcal{A}}(t_1(b_L)).
\end{align}
After $t_1(b_L)^{th}$ mini-slot, the MGG schedule must schedule the remaining users in $\mathcal{A}$. The total number of users scheduled in $\mathbf{f}^G$ is at least $t_1(b_L)+|\mathcal{A}|-N^G_{\mathcal{A}}(t_1(b_L))\geq t_1(b_L)+|\mathcal{A}|-N^{*}_{\mathcal{A}}(t_1(b_L))$.

\section{Proof of Lemma \ref{lem:nonnegative}}\label{APP_12}
\begin{definition}(\textit{Available rate})
	Let $S(t)$ denote the ``available rate'' in $t^{th}$ mini-slot:
	\[S(t+1)=S(t)+\begin{cases}
	t_1(b_j)-t_2(c_j)&\text{if $u^G_t=c_j$},\\
	-(t_2(a_j)-t_1(a_j))&
	\text{if $u^G_t=a_j$}.
	\end{cases}\]
	where the initial value $S(0)=0$ and $u^G_t$ is the $t^{th}$ element in the MGG schedule $\mathbf{f}^G$.
\end{definition}

Start with the first scheduled user in $\mathbf{f}^G$, if we encounter with a user from $\mathcal{C}$, then the ``available rate'' will be added $t_1(b_j)-t_2(c_j)$ more rates offered by users in $\mathcal{C}$. Otherwise, the ``available rate'' will be deducted by ``$\mathcal{A}$ loss rate'' $t_2(a_j)-t_1(a_j)$. In general, $S(t+1)$ is the sum of available rate of queue-length no smaller than $Q_{u^G_t}$. The definition of $S(t)$ allows us to decouple the queue-length from its rate, and to evaluate (\ref{equ:g+E_opt}), we only need to compare the ``available rate" and ``$\mathcal{A}$ loss rate". If for any $1\leq t\leq K$, $S(t)$ is always non-negative, then the R. H. S. of (\ref{equ:g+E_opt}) is also non-negative. Consider each $t$ such that $u^G_t \in\mathcal{A}$, $S(t+1)\geq 0$ means the sum of available rate received by users with  queue-length higher than $Q_{u^G_t}$ is larger than the ``$\mathcal{A}$ loss rate'' on user $u^G_t$. That is to say, for each $a_i$, there will be sufficiently many rate offered by users in $\mathcal{C}$ which have longer queue-length than $Q_{a_i}$. It is sufficient to show that the R. H. S. of (\ref{equ:g+E_opt}) is non-negative. 

On the other hand, we can rewrite the recursion formula of $S(t)$ as:
\[S(t+1)=S(t)+\begin{cases}
t_1(b_j)-t_2(c_j)&\text{if $u^G_t=c_j$},\\
t_1(a_j)-t_2(a_j)&
\text{if $u^G_t=a_j$}.
\end{cases}\]
Start from $t=1$, for each user $u^G_t$, $S(t)$ increments by the time difference of scheduling the same user or the corresponding user under mapping $h$.
Thus, $S(t)$ is actually the difference between the sum of $t$ different timestamps in MaxWeight schedule and the sum of $t$ consecutive timestamps from $1, 2, \cdots$ up to $t$. The later sum is the minimum of the sum of $t$ different timestamps, hence $S(t)\geq 0$ holds for any $1\leq t\leq K$. 

\section{Proof of Lemma \ref{lem:induction}}\label{APP:L9}
We use mathematical induction to prove this lemma.

\textbf{\emph{Base Case:}}
If $m=1$, it is the trivial case, since $\epsilon_i^{K_1}=0$.

If $m=2$, we have:
\begin{align}
\frac{\epsilon^{K_2}_i}{w_i(\mathbf{f}^G_{K_2})}=\frac{Q_{u_1}}{Q_{u_1}(T_2+t_2(u_2)-t_2(u_1)-1)+Q_{u_2}T_2}.
\end{align}

We know that $Q_{u_2}T_2\geq Q_{u_1}$ and $T_2\geq 1$, $t_2(u_2)-t_2(u_1)\geq 1$. Hence $Q_{u_1}(T_2+t_2(u_2)-t_2(u_1)-1)+Q_{u_2}T_2\geq 2Q_{u_1}$ and  $\frac{\epsilon^{K_2}_i}{w_i(\mathbf{f}^G_{K_2})}\leq 1/2$.

\textbf{\emph{Inductive hypothesis:}}
Assume the lemma holds for $m$ users from group $i$, i.e., $\frac{\epsilon^{K_m}_i}{w_i\left(\mathbf{f}^G_{K_m}\right)}\leq 1/2$.

\textbf{\emph{Inductive step:}}
consider the case where we have $m+1$ users from group $i$ ($Q_{u_1}\geq Q_{u_2}\cdots \geq Q_{u_{m+1}}$). $T_{m+1}$ must satisfy:
\begin{numcases}{}
Q_{u_{m+1}}T_{m+1}\geq\sum^m_{j=1}Q_{u_j}. \label{equ:case1}\\
\begin{split}
Q_{u_{m+1}}(T_{m+1}-1)<\sum_{j=1}^mQ_{u_j}.\label{equ:case2}\\
\end{split}
\end{numcases}

User $u_1, u_2, \cdots, u_m$ will determine $T_m$:
\begin{numcases}{}
Q_{u_m}T_m\geq\sum_{j=1}^{m-1}Q_{u_j}.\label{equ:case3}\\
\begin{split}
Q_{u_m}(T_m-1)<\sum_{j=1}^{m-1}Q_{u_j}.\label{equ:case4}\\
\end{split}.
\end{numcases}

We then evaluate $\epsilon^{K_{m+1}}_i/w_i(\mathbf{f}^G_{K_{m+1}})$:
\ifreport
\begin{align}
\frac{\epsilon^{K_{m+1}}_i}{w_i\left(\mathbf{f}^G_{K_{m+1}}\right)}=\frac{\epsilon_i^{K_m}+\overbrace{\sum_{j=1}^mQ_{u_j}}^{\text{Additional extra weight by adding $u_{m+1}$}}}{w_i\left(\mathbf{f}^G_{K_m}\right)+\underbrace{Q_{u_{m+1}} T_{m+1}}_{\text{Additional actual weight on $u_{m+1}$}}+\underbrace{\sum_{j=1}^m Q_{u_j}\left(K_{m+1}-K_{m}-1\right)}_{\text{Additional actual weight on $u_1,\cdots, u_m$}} }.
\end{align}
\fi

Given the inductive hypothesis, it suffices to show
\begin{align}
2\sum_{j=1}^mQ_{u_j}\leq Q_{u_{m+1}} T_{m+1}+\sum_{j=1}^m Q_{u_j}\left(K_{m+1}-K_{m}-1\right).
\end{align}
From (\ref{equ:case1}), we already know $\sum^m_{j=1}Q_{u_j}\le Q_{u_{m+1}}T_{m+1}$. We only need to show $\sum^m_{j=1}Q_{u_j}\le \sum_{j=1}^m Q_{u_j}\left(K_{m+1}-K_{m}-1\right)$, equivalently, $K_{m+1}-K_{m}-1\geq 1$.
By definition, $K_{m+1}=t_2(u_{m+1})+T_{m+1}\geq t_2(u_m)+1+T_{m+1}$, $K_m=t_2(u_m)+T_m$. The only thing left is to show $T_{m+1}-T_{m}\geq 1$ $(T_{m+1}>T_{m})$. Suppose  $T_{m}\geq T_{m+1}$, from (\ref{equ:case4}), we know:
\begin{align}
Q_{u_m}T_m<\sum_{j=1}^{m-1}Q_{u_j}+Q_{u_m}=\sum_{j=1}^m Q_{u_j}.
\end{align}
Then,
\begin{align}
Q_{u_{m+1}}T_{m+1}\leq Q_{u_m}T_m<\sum_{j=1}^mQ_{u_j}.\label{equ:condiaict}
\end{align}
(\ref{equ:condiaict}) contradicts (\ref{equ:case1}), therefore, $T_{m+1}>T_m$, $T_{m+1}-T_m\geq 1$, Lemma \ref{lem:induction} holds. 

\section{Proof of Lemma \ref{lem:weight_domin}}\label{APP_L10}
Assume a scheduling policy $\mathbf{f}$ under Half-duplex will schedule $m$ users $u_1, \cdots, u_m$. The transmission starts at $m+1^{th}$ mini-slot, and each user will receive a rate of $K-m$. The total weight $w_{HD}(\mathbf{f})$ can be calculated as:
\begin{align}
w_{HD}(\mathbf{f})=\sum_{i=1}^m Q_{u_i}\left(K-m\right)
\end{align}

To maximize $w_{HD}(\mathbf{f})$, we know that given $m$, the scheduler must choose $m$ users with the longest queue-length. We only need to decide the optimal value for $m$. Start from $m=1$, we can evaluate the weight $w_{HD}(\mathbf{f})$ for each $m$, the MaxWeight policy chooses the smallest $m^*$ such that adding one more user to the schedule decreases the total weight (negative ``marginal gain"). Assume $Q_{u_1}\geq Q_{u_2}\cdots \geq Q_{u_N}$, define $m^*$ as follows: 
\begin{align}
m^*\triangleq \min\{m: Q_{u_{m+1}}(K-m-1)< \sum_{i=1}^{m}Q_{u_i}\}.\label{equ:th3_1}
\end{align}

Consider the MGG schedule $\mathbf{f}^G=(u^g_1, \cdots, u^g_{\Omega},0,\cdots,0)$. Let $m'$ denote the largest integer such that for any $i\le m'$, $u^g_i=u_i$ holds. That means the MGG schedule also chooses $m'$ users with longest queue-length for the first $m'$ mini-slots. Since $m'$ is the largest integer, it must satisfy
\begin{align}
\Delta_{u_{m'+1}}^{\mathbf{f}^G,m'+1}=Q_{u_{m'+1}}(K-m'-1)-\sum_{i=1}^{m'}Q_{u_i}\mathbf{1}_{\{g(u_i)=g(u_m')\}}<0.\label{equ:th3_2}
\end{align}
which means $u_{m'+1}$ cannot be added into the MGG schedule after $u_1, u_2, \cdots, u_{m'}$.

Suppose $m'<m^*$, then we have:
\begin{align}
Q_{u_{m'+1}}(K-m'-1)< \sum_{i=1}^{m'}Q_{u_i}\mathbf{1}_{\{g(u_i)=g(u_m')\}}\leq\sum_{i=1}^{m'}Q_{u_i}.\label{equ:th3_3}
\end{align}
(\ref{equ:th3_3}) implies that there exists $m'<m^*$, such that the inequality in (\ref{equ:th3_1}) still holds. This contradicts the definition of $m^*$. Thus, $m'\geq m^*$.

Let $\tilde{\mathbf{f}}=(u_1, u_2,\cdots, u_{m^*}, 0 ,\cdots,0)$ and $\tilde{\mathbf{f'}}=(u_1, u_2, \cdots, u_{m'}, 0, \cdots, 0)$, $\tilde{\mathbf{f}}$ and $\tilde{\mathbf{f'}}$ are intermediate schedules produced by the MGG Algorithm after $m^{*}$ and $m'$ iterations. Hence, we have:
\begin{align}
w(\tilde{\mathbf{f'}})-w(\tilde{\mathbf{f}})=\sum_{i=m^*+1}^{m'}\Delta^{\mathbf{f}^G, i}_{u^g_i}\geq 0\label{equ:th3_5}
\end{align}
The weight difference is just the sum of marginal gain for adding each user of $u^g_{m^*+1}, u^g_{m^*+2},\cdots, u^g_{m'}$. For the same scheduling vector $\tilde{\mathbf{f}}$, we have $w(\tilde{\mathbf{f}})\geq w_{HD}(\tilde{\mathbf{f}})$. Combining with (\ref{equ:th3_5}), we have $w(\mathbf{f}^G)\geq w(\tilde{\mathbf{f'}})\geq w(\tilde{\mathbf{f}}) \geq w_{HD}(\tilde{\mathbf{f}})=w^*_{HD}$, which concludes the proof of Lemma \ref{lem:weight_domin}.

\section{Proof of Lemma \ref{lem:lyapunov}}\label{APP_11}
Since $\boldsymbol{\lambda}\in\text{int }\Lambda^*_{HD}$, there exists a $\delta>0$ and a mean service rate vector $\mu^{prob}\in \Lambda^*_{HD}$ that is achievable by a probabilistic policy $P^{prob}$ such that:
\begin{align}
\mu^{prob}_i\geq \lambda_i+\delta,  \text{   for all }i.
\end{align}

Define the quadratic Lyapunov function $V(\cdot)$ as:
\begin{align}
V(\mathbf{Q})=\frac{1}{2}\sum_{i=1}^N Q^2_i.
\end{align}

Let $\Delta_t V(\mathbf{q})$ denote the mean drift of $V(\cdot)$ for any $\mathbf{Q}[t]=\mathbf{q}$:
\begin{align}
\Delta_t V(\mathbf{q})=\mathbb{E}[V(\mathbf{Q}[t+1])-V(\mathbf{Q}[t])|\mathbf{Q}[t]=\mathbf{q}].
\end{align}
$\Delta_t V(\mathbf{q})$ can be upper bounded by:
\begin{align}
\Delta_t V(\mathbf{q})&\leq \sum_{i=1}^N q_i\mathbb{E}\left[A_i-R^{\mathbf{f}^G}_i|\mathbf{Q}[t]=\mathbf{q}\right]+\frac{1}{2}\sum_{i=1}^N  \mathbb{E}\left[\left(A_i-R^{\mathbf{f}^G}_i\right)^2|\mathbf{Q}[t]=\mathbf{q}\right].\nonumber\\
&\leq \sum_{i=1}^N q_i\left(\lambda_i-\mathbb{E}\left[\left(R^{\mathbf{f}^G}_i\right)^2|\mathbf{Q}[t]=\mathbf{q}\right]\right)+B.
\end{align}
as long as $\mathbb{E}[A^2_i]<\infty$ and $\mathbb{E}\left[\left(R^{\mathbf{f}^G}_i\right)^2\right]<\infty$, the second expectation is bounded by a constant $B<\infty$.

Let $\mathbf{f}_{HD}^*[t]$ denote the MaxWeight scheduling vector in time-slot $t$ under Half-duplex, we have:
\begin{align}
\mathbf{f}_{HD}^*[t]=\argmin_{\mathbf{f}\in\mathcal{F}_{HD}}\left(-\sum_{i=1}^N Q_i[t]R^{\mathbf{f}}_i[t]\right)=\argmax_{\mathbf{f}\in\mathcal{F}_{HD}}\left(\sum_{i=1}^N Q_i[t]R^{\mathbf{f}}_i[t]\right).
\end{align}

Due to its minimizing nature, we have
\begin{align}
-\sum_{i=1}^N q_i \mathbb{E}[R_i^{\mathbf{f}_{HD}^*}|\mathbf{Q}[t]=\mathbf{q}]&\leq -\sum_{i=1}^N q_i \mathbb{E}[R_i^{prob}|\mathbf{Q}[t]=\mathbf{q}]\nonumber\\
&\leq -\sum_{i=1}^N q_i \mathbb{E}[R_i^{prob}]=-\sum_{i=1}^N q_i\mu_i^{prob}\leq \sum_{i=1}^N q_i(\lambda_i+\delta).\label{equ:1}
\end{align}

On the other hand, from Lemma \ref{lem:weight_domin}, we know the weight $w(\mathbf{f}^G)$ dominates $w^*_{HD}$ under any queue-length vector $\mathbf{q}$. Therefore, 
\begin{align}
-\sum_{i=1}^Nq_i\mathbb{E}[R_i^{\mathbf{f}^G}|\mathbf{Q}[t]=\mathbf{q}]\leq -\sum_{i=1}^N q_i \mathbb{E}[R_i^{\mathbf{f}_{HD}^*}|\mathbf{Q}[t]=\mathbf{q}].\label{equ:2}
\end{align}

Combining (\ref{equ:1}) and (\ref{equ:2}), we have:
\begin{align}
-\sum_{i=1}^Nq_i\mathbb{E}[R_i^{\mathbf{f}^G}|\mathbf{Q}[t]=\mathbf{q}]\leq \sum_{i=1}^N q_i(\lambda_i+\delta).
\end{align}

Use this bound in the Lyapunov drift upper bound:
\begin{align}
\Delta_tV(\mathbf{q})&\leq \sum_{i=1}^N q_i (\lambda_i-\mathbb{E}[R^{\mathbf{f}^G}_i|\mathbf{Q}[t]=\mathbf{q}])+B\leq \sum_{i=1}^N q_i(\lambda_i-(\lambda_i+\delta))+B\nonumber\\
&\leq -\delta\sum_{i=1}^N q_i +B.
\end{align}

Applying Foster-Lyapunov Theorem, we know the queueing network is stable.
\end{appendices}

\end{document}